%% file: main_after_submission.tex
\documentclass[journal]{new-aiaa} 
\usepackage[utf8]{inputenc}
\usepackage{graphicx}
\usepackage{amsmath,bm,bbm}
\usepackage[version=4]{mhchem}
\usepackage{siunitx}
\usepackage{multicol}
\usepackage{amsthm}
\usepackage{longtable,tabularx}
\usepackage{diagbox}
\usepackage{multirow}
\usepackage{algorithm}
\usepackage{algpseudocode}
\input{command}

\graphicspath{{Figures/}}
\newtheorem{theorem}{Theorem}
\newtheorem{corollary}{Corollary}[theorem]
\newtheorem{rem}{Remark}
\newtheorem{proposition}{Proposition}

\usepackage{ifthen}
\usepackage{comment}
\title{Optimal Ground-to-Air Interception with \\ Time-Varying Acceleration Bounds}

\author{Or Nahum\footnote{Ph.D. Candidate, The Stephen B. Klein Faculty of Aerospace Engineering, Technion, \href{mailto:nahum.or@campus.technion.ac.il}{nahum.or@campus.technion.ac.il}.} and Vitaly Shaferman\footnote{Associate Professor, The Stephen B. Klein Faculty of Aerospace Engineering, Technion, \href{mailto:vitalysh@technion.ac.il}{vitalysh@technion.ac.il} (Corresponding Author).}}
\affil{Technion, Israel Institute of Technology, Haifa, Israel, 3200003}

\begin{document}

\maketitle

\begin{abstract}
This paper proposes novel optimal-control-based guidance laws for ground-to-air missiles with time-varying acceleration bounds. In such engagements, as the missile climbs in altitude, its acceleration bound decreases, which may lead to acceleration saturation and significant miss distances if not explicitly accounted for. The proposed guidance laws incorporate hard acceleration command constraints directly into a linear-quadratic optimal-control framework, in contrast to conventional unbounded or softly constrained approaches. Analytically based guidance laws are developed for linear zero-order and first-order strictly proper missile dynamics with arbitrary-order linear target dynamics.
Unlike the constant hard-bound case with minimum-phase missile dynamics, time-varying acceleration command bounds permit an initial unsaturated interval in which the proposed guidance laws can anticipate future saturation and reshape the acceleration profile accordingly. This enables earlier maneuvers when the missile possesses greater low-altitude maneuverability, fundamentally altering the structure of the optimal solution.
The proposed approach is evaluated in nonlinear simulations and compared with equivalent unbounded and softly constrained optimal guidance laws. The results demonstrate substantially improved interception performance under saturation, reduced tuning requirements compared to softly constrained guidance laws, and enhanced capability in challenging engagement scenarios.
\end{abstract}

\footnotetext[0]{Presented as Paper 2024-1986 at the AIAA SciTech Forum, Orlando, FL, 8 - 12 January 2024.}

\section{Introduction}  \label{Introduction}
\lettrine{A}{cceleration} saturation is a challenging problem in missile guidance. It occurs when the missile acceleration or acceleration command exceeds the missile’s achievable capability. Such saturation may result in substantial miss distances and eventual failure to intercept the target. From a control perspective, saturation effectively ``opens'' the closed guidance loop, since the commanded acceleration can no longer be accurately tracked. Most classical guidance laws, including Proportional Navigation (PN) \cite{yuan1948homing} and Augmented Proportional Navigation (APN) \cite{garber1968optimum}, do not explicitly account for acceleration saturation.

Optimal control is a particularly appealing framework for deriving missile guidance laws, because it enables a systematic minimization of prescribed performance criteria. Multiple objectives, such as miss distance, control effort, and terminal interception-angle error, can be incorporated into the guidance law design. A well-known optimal-control-based guidance law is the Optimal Guidance Law (OGL) \cite{cottrell1971optimal}, derived for a linear-quadratic problem formulation with first-order, strictly proper missile dynamics, which minimizes both the miss distance and the control effort. Using a similar framework, \citet{shaferman2008linear} derived a guidance law that additionally minimizes the interception-angle error relative to a prescribed terminal-angle requirement. However, these guidance laws do not explicitly account for acceleration saturation.

Optimal-control-based missile guidance laws that explicitly account for acceleration constraints have been developed in the literature to mitigate saturation effects. Two main types of acceleration bounds have been considered: constant bounds and altitude-dependent bounds, in which the allowable acceleration decreases with altitude. The latter more accurately represents ground-to-air interception scenarios, in which the missile climbs in altitude, and its maneuvering capability decreases throughout the engagement due to the decrease in density.

\citet{rusnak1990optimal} derived an optimal guidance law for linear missile dynamics with a quadratic cost on the control effort and terminal miss distance under a ``hard'' constant acceleration command constraint. They showed that for minimum-phase missiles, the unbounded optimal guidance law remains optimal even when the acceleration command is bounded. However, for non-minimum-phase missiles, the unbounded guidance law becomes suboptimal in the presence of saturation. Nevertheless, due to the complexity of the corresponding bounded optimal solution, the authors recommended using the unbounded guidance law in practice. A subsequent work by \citet{rusnak1991optimal} reached similar conclusions for the same problem formulation with higher-order missile dynamics.

In a later paper, \citet{rusnak2008guidance} extended the previous works to a Linear Quadratic Differential Game (LQDG) with constant missile and target acceleration command constraints. For minimum-phase missile and target dynamics, implicit guidance-law expressions were derived and solved numerically. Unlike optimal guidance laws, which assume the target maneuver is known and, therefore, involve a one-sided optimization, Differential Games Guidance Laws (DGGLs) model the target maneuver as unknown and are consequently formulated as two-sided optimizations \cite{mishley2022linear}.

Another differential game formulation with constant acceleration constraints is a bounded-control, zero-sum game with the miss distance as the cost function. \citet{gutman79} derived the solution for arbitrary-order linear missile and target dynamics. As a representative example, he explicitly solved the problem for ideal target dynamics and first-order, strictly proper missile dynamics, resulting in the well-known DGL0 guidance law. \citet{shinar81} considered a similar formulation for first-order, strictly proper missile and target dynamics and obtained the celebrated DGL1 guidance law.

\citet{weiss2015optimal} proposed an alternative optimal-control-based approach to ``softly'' address constant acceleration command limits. In addition to quadratic penalties on the miss distance and control effort, the cost function included a running penalty on the deviation of the instantaneous control from its average value, thereby reducing control variations throughout the engagement. This modification can mitigate acceleration saturation by encouraging smoother acceleration command profiles. Intuitively, penalizing deviations from the average control favors a more uniform control history, thereby reducing the peak control magnitude during the engagement.

All the guidance laws discussed thus far considered constant acceleration command limits. However, in many practical scenarios, a time-varying bound provides a more accurate representation of the missile’s maneuvering capability. \citet{ben2003new} proposed a soft constraint approach based on an exponential time-to-go weighting of the control-effort running cost. The motivation stems from the fact that the maneuverability of a ground-to-air missile typically deteriorates with altitude due to the reduction in atmospheric density. Since the density profile can be approximated exponentially as a function of altitude, the suggested weighting function is physically well motivated. The derived guidance law, therefore, tends to command larger accelerations earlier in the engagement, when the missile possesses greater maneuverability.

\citet{taub2013intercept} also proposed a soft acceleration constraint implemented through a time-varying weighting of the control-effort cost. Iterative tuning was required to achieve small miss distances while avoiding saturation. In addition, the proposed guidance law enforced a terminal interception-angle constraint.
\citet{shima2002time} derived a DGGL with hard time-varying acceleration command bounds on both the missile and the target, using the miss distance as the cost function. However, the resulting guidance law did not minimize the control effort, which is often critical for preserving interceptor maneuverability in ground-to-air engagements. Furthermore, as a DGGL, the formulation was inherently conservative and did not exploit possible information regarding the future target maneuver.

This paper proposes bounded optimal guidance laws for ground-to-air missiles with hard time-varying acceleration command constraints. The guidance problem is formulated as a linear-quadratic optimal-control problem with time-varying acceleration command bounds arising from altitude-dependent acceleration limits. The proposed guidance laws constitute bounded extensions of the classical unbounded OGL.

Ground-to-air engagements are characterized by significant variations in the missile’s maximum achievable acceleration due to the altitude-dependent reduction in atmospheric density. Consequently, neglecting acceleration saturation may prevent successful interception even in engagements that would otherwise be interceptable. Unlike the constant hard-bound case with minimum-phase missile dynamics analyzed in \cite{rusnak1990optimal,rusnak1991optimal}, time-varying acceleration command bounds may induce multiple transitions between unsaturated and saturated arcs, typically giving rise to an initial unsaturated interval prior to saturation in ground-to-air interception scenarios. The existence of this interval enables the proposed guidance laws to anticipate the upcoming constraint activation and reshape the acceleration command profile accordingly. Physically, this initial unsaturated interval arises from the higher air density at lower altitudes, which permits larger acceleration command limits and enables the missile to maneuver more aggressively before entering the saturated region. This fundamentally alters the structure of the optimal solution.
Furthermore, unlike softly constrained approaches, the proposed framework does not require iterative tuning of weighting parameters. Two guidance laws are developed: one for zero-order missile dynamics and another for first-order strictly proper missile dynamics.

The proposed approach is evaluated through nonlinear simulations and compared with equivalent unbounded and softly constrained optimal guidance laws, demonstrating the benefits of explicitly accounting for hard time-varying acceleration command constraints in challenging engagement scenarios. An earlier version of the proposed approach was presented in \cite{nahum2024optimal} and is extended here.

The remainder of this paper is organized as follows. Section \ref{Model Derivation} presents the model derivation, while the guidance problem formulation and the order reduction are introduced in Section \ref{Guidance Problem Formulation and Order Reduction}. The proposed Bounded Optimal Guidance Laws (BOGLs) are derived in Section \ref{Bounded Optimal Guidance Law}, followed by the implementation algorithms in Section \ref{Algorithm Implementation}. Simulation results are provided in Section \ref{Simulations}, and concluding remarks are given in Section \ref{Conclusion}. \ref{Appendix_Derivation} provides detailed derivations of the relevant transition-matrix components, the zero-effort miss, the reduced-order dynamics, and the required integrals. \ref{Appendix_Proof} establishes sufficient conditions for the existence of at most two transitions between unsaturated and saturated arcs, while \ref{Appendix_Algo} summarizes the implementation algorithms for the zero-order BOGL.

\section{Model Derivation} \label{Model Derivation}

The engagement geometry is shown in \figref{Engagement Geometry}. The missile and the target are denoted by the subscripts $M$ and $T$, respectively. The acceleration, velocity, and flight-path angle are denoted by $a$, $v$,
and $\gamma$, respectively. The relative displacement between the target and the missile perpendicular to the initial Line-Of-Sight (LOS),  denoted by $LOS(0)$, is $y$. The angle between the LOS and the inertial $X$ direction is $\sigma$, and the distance between the target and the missile is $r$. 
\begin{figure}[hbt!] 
\centering
\includegraphics[width=0.65\textwidth]{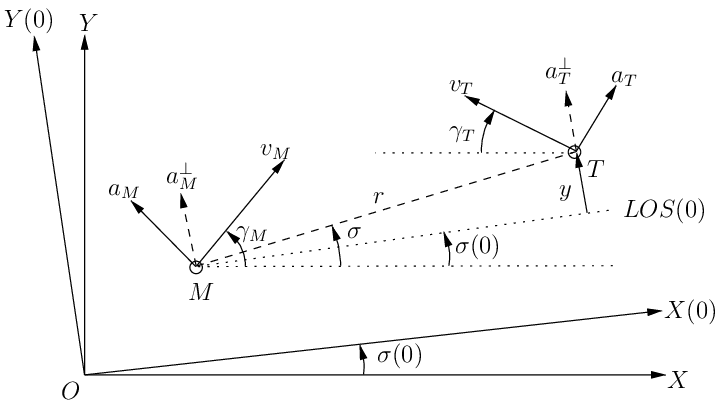}
\caption{The engagement geometry.}
\label{Engagement Geometry}
\end{figure}

The target and the missile accelerations perpendicular to $LOS(0)$ are $a_{T}^{\perp}$ and $a_{M}^{\perp}$, respectively, and are given by
\begin{subequations}
\begin{align} \label{K_def}
a_{T}^{\perp} &= K_{T} a_{T}, \quad \ K_{T}=\cos\left[\gamma_{T}(0)+\sigma(0)\right]\\
a_{M}^{\perp} &= K_{M} a_{M}, \quad\ K_{M}=\cos\left[\gamma_{M}(0)-\sigma(0)\right] 
\end{align}
\end{subequations}

\subsection{Nonlinear Kinematics} \label{Non-Linear Kinematics}
The nonlinear kinematics equations are expressed in polar coordinates, $r$ and $\sigma$. 
The radial speed along the LOS is
\begin{equation} \label{vr}
\dot r  = v_{r} = - \left[ v_{M} \cos(\gamma_M-\sigma) +v_{T} \cos(\gamma_T+\sigma)  \right]
\end{equation}
and the relative speed perpendicular to the LOS is
\begin{equation} \label{v_perp}
 r \dot \sigma  = v_{\perp} = - v_{M} \sin(\gamma_M-\sigma) +v_{T} \sin(\gamma_T+\sigma)   \quad
    \Longrightarrow \quad  \dot \sigma  = \frac{v_{\perp}}{r}
\end{equation}
The target and missile are assumed to move at constant speeds, and their flight-path angle rates are
\begin{equation} \label{path_angles}
 \dot \gamma_{M}  = \frac{a_M}{v_M}, \quad \dot \gamma_{T}  = \frac{a_T}{v_T}
\end{equation}
Arbitrary order, Linear Time-Invariant (LTI) dynamic models are assumed for both the missile and the target. The dynamic models can be described by the following state-space representations
\begin{subequations} \label{state_space_target&missile}
\begin{align}
\dot {\mathbf{x}}_{M} &= \mathbf{A}_{M} \mathbf{x}_{M}+\mathbf{b}_{M} u, \quad   a_M=\mathbf{c}_{M}^{T} \mathbf{x}_{M}+d_{M}u \\
\dot{\mathbf{x}}_{T} &= \mathbf{A}_{T} \mathbf{x}_{T}+\mathbf{b}_{T} w, \quad  a_T=\mathbf{c}_{T}^{T} \mathbf{x}_{T}+d_{T}w
\end{align}
\end{subequations}
where ${\mathbf{x}}_{M}$ and ${\mathbf{x}}_{T}$ are the state vectors of the missile's and the target's internal state variables, respectively, and the commanded accelerations of the missile and the target are $u$ and $w$, respectively.

\subsection{Linearized Kinematics for Guidance Law Derivation}
In this subsection, a linearized model is developed to simplify the analytical derivation of the guidance laws. The derivation is valid when the missile and the target are on a collision triangle or when their deviations from it are minor. Since the missile and the target are not necessarily near or on a collision course, the nonlinear dynamics will be linearized at every time step for the guidance law implementation (extended linearization). 
Let the state vector be
\begin{equation} \label{state_vector}
\mathbf{x}= 
\left[y \quad \dot{y} \quad \mathbf{x}_{T}^{T} \quad \mathbf{x}_{M}^{T}\right]
^{T}
\end{equation}
The Equation Of Motion (EOM) perpendicular to the LOS is
\begin{equation} \label{y_ddot}
    \ddot{y} =  K_{T}a_{T}-K_{M}a_{M}=K_{T}(\mathbf{c}_{T}^{T} \mathbf{x}_{T}+d_{T}w)-K_{M}(\mathbf{c}_{M}^{T}\mathbf{x}_{M}+d_{M}u)
\end{equation}
Combining \TWOeqsref{state_space_target&missile}{y_ddot} we obtain the following state-space model
\begin{equation}  \label{state_space}
    \dot{\mathbf{x}} = \mathbf{Ax}+\mathbf{b}u+\mathbf{c}w
\end{equation}
where
\begin{equation} \label{state space terms}
    \mathbf{A}=
    \begin{bmatrix}
        0 & 1 & [\mathbf{0}] & [\mathbf{0}] \\
        0 & 0 & K_{T}\mathbf{c}_{T}^{T} & -K_{M}\mathbf{c}_{M}^{T} \\
        [\mathbf{0}] & [\mathbf{0}] & \mathbf{A}_{T} & [\mathbf{0}] \\
        [\mathbf{0}] & [\mathbf{0}] & [\mathbf{0}] & \mathbf{A}_{M} \\
    \end{bmatrix}
    ,\quad     \mathbf{b}=
    \begin{bmatrix}
        0  \\
        -K_{M}d_{M} \\
        [\mathbf{0}] \\
        \mathbf{b}_{M} \\
    \end{bmatrix}
        ,\quad     \mathbf{c}=
    \begin{bmatrix}
        0  \\
        K_{T}d_{T} \\
        \mathbf{b}_{T} \\
        [\mathbf{0}] \\
    \end{bmatrix}
\end{equation}
and $[\mathbf{0}]$ is the zero matrix with appropriate dimensions.
The time-to-go, which is the time remaining from the current moment until the end of the scenario, is approximated by
\begin{equation} \label{tgo}
    t_{go}\dfn t_{f}-t \approx \frac{r}{v_{c}}
\end{equation}
where $t_{f}$ is the final time and $v_{c}=-v_{r}$ is the closing speed.

\section{Guidance Problem Formulation and Order Reduction} \label{Guidance Problem Formulation and Order Reduction}
We choose a  quadratic cost function containing both a terminal cost on the miss distance and a running cost on the control effort
\begin{equation} \label{spesific cost}
    \mathcal{J}=\frac{\alpha}{2}x_{1}^{2}(t_{f})+\frac{1}{2} \int_{t}^{t_{f}} u^{2}(\xi) \,\dd\xi
\end{equation}
where $\alpha$ is a positive weighting parameter. Notice that $\alpha \rightarrow \infty$ with no acceleration bounds yields a perfect interception. 
The cost function should be minimized by the bounded controller $u$. 
In the derivation, the target's acceleration command, $w$, is assumed known to the missile. Since the cost function presented in \eqref{spesific cost} contains only the first state at the final time, and the dynamics are linear, the problem can be reduced to a scalar problem. The method that will be used is often called the ``Terminal projection'' method \cite{bryson2018applied}.
Let us define a new scalar variable
\begin{equation} \label{Z}
    Z(t)=\mathbf{D\Phi}(t_{f},t)\mathbf{x}+\mathbf{D} \int_{t}^{t_{f}} \mathbf{\Phi}(t_{f},\tau) \mathbf{c} w \,\dd\tau, \quad \mathbf{D}= 
    \begin{bmatrix}
    1 & 0 & [\mathbf{0}] & [\mathbf{0}]
    \end{bmatrix}
\end{equation}
where $\mathbf{\Phi}(t_{f},t)$ is the transition matrix associated with the dynamics presented in \eqref{state_space}. Since the dynamics are LTI, the transition matrix is a function of the subtraction of its two arguments, i.e., $\mathbf{\Phi}(t_{f},t)=\mathbf{\Phi}(t_{f}-t)=\mathbf{\Phi}(\tgo)$. 
Differentiating \eqref{Z} with respect to time and substituting \eqref{state_space} and the following transition matrix property \cite{gutman2005applied}
\begin{equation} \label{ODE_Phi}
    \dot{\mathbf{\Phi}}(t_{f},t) = -\mathbf{\Phi}(t_{f},t)\mathbf{A}, \quad \mathbf{\Phi}(t,t)=\mathbf{I}
\end{equation}
we obtain
\begin{equation} \label{Z_dot}
        \dot{Z}(t)  =\varphi (t_{go})u, \quad \varphi (t_{go})\triangleq \mathbf{D} \mathbf{\Phi} (t_{go})\mathbf{b}
\end{equation}
where $\varphi(t_{go})$ depends on the dynamics of the problem and requires the derivation of $\mathbf{D}\mathbf{\Phi} (t_{go})$, i.e., the first row of  $\mathbf{\Phi} (t_{go})$.
It is evident from \eqref{tgo} that for a general function $\xi$ the following statement holds
\begin{equation} \label{Derivative}
    \frac{\dd\xi}{\dd t_{go}}=\frac{\dd\xi}{\dd t} \cdot \frac{\dd t}{\dd t_{go}} = -\dot{\xi}
\end{equation}
Substituting \eqref{Derivative} into \eqref{ODE_Phi} and multiplying by $\mathbf{D}$ from the left-hand-side, we obtain the following Ordinary Differential Equation (ODE)
\begin{equation} \label{Phi_ode}
    \mathbf{D}\frac{\dd\mathbf{\Phi}(\tgo)}{\dd t_{go}} = \mathbf{D\Phi}(\tgo)\mathbf{A}, \quad 
    \mathbf{\Phi}(t_{go}=0)=\mathbf{I}
\end{equation}
which can be solved in closed form, yielding 
\begin{equation} \label{DPhi}
      \mathbf{D}\mathbf{\Phi} =
      \left[1 \; \; t_{go} \; \; \bm{\phi}_{1T}^{T} \; \; \bm{\phi}_{1M}^{T} \right], \quad 
      \bm{\phi}_{1T}^{T} = K_{T} \mathscr{L}_{t_{go}}^{-1} \left[ \frac{1}{s^2} \mathbf{c}_{T}^{T} \left(s\mathbf{I}-\mathbf{A}_T \right)^{-1} \right], \quad 
    \bm{\phi}_{1M}^{T} = -K_{M} \mathscr{L}_{t_{go}}^{-1} \left[ \frac{1}{s^2} \mathbf{c}_{M}^{T} \left(s\mathbf{I}-\mathbf{A}_M \right)^{-1} \right] 
\end{equation}
where $\mathscr{L}_{t_{go}}^{-1}$ is the inverse Laplace transform in the $t_{go}$ domain. For brevity, the dependence of the transition matrix on $\tgo$ is omitted in the remainder of the paper unless needed for clarity.
Substituting \eqref{DPhi} into \eqref{Z_dot} we get
for zero-order and first-order, strictly proper missile dynamics, respectively
\begin{subequations} \label{phi0_phi1}
\begin{align} \label{phi0}
    &\varphi_{0}(\tgo)=-K_{M}\tgo\\  \label{phi1}
     &\varphi_{1}(\theta)=-K_{M}\tau_{M}\psi(\theta), \quad     \psi(\theta) \dfn {\mathrm e}^{-\theta}+\theta-1,  \quad \theta \dfn \frac{t_{go}}{\tau_{M}}
\end{align}   
\end{subequations}
where $\tau_{M}$ is the time constant of the missile.
Substituting \eqsref{phi0_phi1} into \eqref{Z_dot}, we obtain the dynamics of $Z(t)$ for zero-order and first-order, strictly proper missile dynamics.
To improve readability, the derivations of \TWOeqsref{DPhi}{phi0_phi1} are given in \ref{Appendix_Derivation}. 

The cost function in \eqref{spesific cost} can be rewritten
by substituting $t=t_{f}$ in \eqref{Z}, obtaining $Z(t_{f})=x_{1}(t_{f})$, and 
\begin{equation}
    \mathcal{J}=\frac{\alpha}{2}Z^{2}(t_{f})+\frac{1}{2} \int_{t}^{t_{f}} u^{2}(\xi)\,\dd\xi
\end{equation}
The variable $Z$ is known as the Zero-Effort-Miss (ZEM), and it equals the miss distance if the missile acceleration command is nulled from the current time onward and the target performs its assumed maneuver.

\section{Bounded Optimal Guidance Laws} \label{Bounded Optimal Guidance Law}
\subsection{Optimal Controller Derivation}
The Hamiltonian of the problem is
 \begin{equation}
 \mathcal{H} = \frac{1}{2}u^{2}+\lambda  \dot{Z}=\frac{1}{2}u^{2}+\lambda \varphi (t_{go}) u\
\end{equation}
where $\lambda$ is the co-state, which satisfies
 \begin{equation} 
    \dot{\lambda}=-\frac{\partial \mathcal{H}}{\partial Z}=0, \quad 
    \lambda(t_{f})=
    \alpha Z(t_{f}) \quad \Longrightarrow \quad \lambda (t) =  \alpha Z(t_{f})
\end{equation}
The unbounded optimal controller satisfies 
 \begin{equation} \label{u_opt_free}
       u^{*} (t_{go})= \arg \min_{u} \mathcal{H} \quad
    \Longrightarrow \quad \frac{\partial \mathcal{H}}{\partial u}=0 \quad
    \Longrightarrow \quad 
     u^{*} (t_{go})= -\alpha Z(t_{f}) \varphi (t_{go})
\end{equation}
One can derive $\varphi (t_{go})$ and substitute it into the optimal controller's structure in \eqref{u_opt_free} for a specific dynamic model. 

Substituting \eqsref{phi0_phi1} into \eqref{u_opt_free} we obtain the unbounded optimal controller structures for zero-order and first-order, strictly proper missile dynamics, respectively
\begin{subequations}\label{Unbounded Optimal Controllers - zero and First}
\begin{align}\label{Unbounded Optimal Controller - zero}
&u^{*}_{0}(t_{go}) = \mathcal{K}_{0}t_{go}, \quad  \mathcal{K}_{0}=\alpha Z(t_{f})K_{M} \\ \label{Unbounded Optimal Controller - 1st}
&u^{*}_{1}(\theta) = \mathcal{K}_{1}\psi (\theta), \quad  \mathcal{K}_{1}=\alpha Z(t_{f})K_{M} \tau_{M}
\end{align}  
\end{subequations}

In the bounded problem, saturation is explicitly considered. The bounded optimal controller satisfies
\begin{equation} \label{Bounded Optimal controller}
    \bar{u}^{*} (t_{go})= \arg \min_{\left| u\right|\le u_{\max}} \mathcal{H} \quad \Longrightarrow \quad     \bar{u}^* (t_{go})= -u_{\max}(t_{go})\cdot \text{sat} \left[ \frac{\alpha \bar{Z}(t_{f}) \varphi (t_{go})}{u_{\text{max}}(t_{go})} \right], \quad \text{sat}(\xi) =
    \begin{cases}
			\frac{\xi}{\left| \xi \right|}, & \left| \xi \right|>1\\
            \xi, & \left| \xi \right|\le1
	\end{cases}
\end{equation}
where sat is the saturation function and ${u_{\text{max}}(t_{go})}$ is the absolute value of the missile's acceleration command limit, which might be constant or time-dependent. Note that $\bar{Z}(t_{f})$ is the terminal ZEM in the bounded scenario. We specifically marked ${Z}(t_{f})$ in the bounded scenario with a bar to highlight that ${Z}(t_{f})$ in the bounded and unbounded cases might be different.

Substituting \eqsref{phi0_phi1} into \eqref{Bounded Optimal controller}, we obtain the bounded optimal controller structures for zero-order and first-order, strictly proper missile dynamics, respectively
\begin{subequations}\label{Bounded Optimal Controllers - zero and 1st}
\begin{align}\label{Bounded Optimal Controller - 0th} 
&\bar{u}^{*}_{0}(t_{go}) = u_{\max}(t_{go})\cdot \text{sat} \left[ \frac{\bar{\mathcal{K}}_{0}t_{go}}{u_{\text{max}}(t_{go})} \right], \quad  \bar{\mathcal{K}}_{0}=\alpha \bar{Z}(t_{f})K_{M} \\ \label{Bounded Optimal Controller - 1st}
&\bar{u}^{*}_{1}(\theta) = u_{\max}(\theta)\cdot \text{sat} \left[ \frac{\bar{\mathcal{K}}_{1} \psi (\theta)}{u_{\text{max}}(\theta)} \right], \quad  \bar{\mathcal{K}}_{1}=\alpha \bar{Z}(t_{f})K_{M} \tau_{M}
\end{align}
\end{subequations}
 For brevity, we will not carry the dependency of $\varphi$ and $u_{\text{max}}$ on $t_{go}$ in the rest of the paper, unless it is needed to understand the expressions. 
 
As the introduction noted, the missile's acceleration limit in a ground-to-air engagement is typically altitude-dependent. We assume that the missile's acceleration limit as a function of altitude is given.
Assuming the missile and the target are on a collision course, or their deviations from it are small, then the altitude change rate is approximately constant. Therefore, the acceleration limit as a function of the missile's altitude behaves similarly to the acceleration limit as a function of time. 
To simplify the problem, the magnitude of the missile's acceleration command limit is approximated as an $n$-th order polynomial function of the time-to-go
\begin{equation}\label{umax(tgo)}
u_{\text{max}}\left( t_{go} \right) = \sum_{i=0}^{n} \tilde{b}_{i} t_{go}^{i}, \quad \tilde{b}_{0}>0
\end{equation}
where $\tilde{b}_{0}>0$ since the missile is assumed to always have some maneuverability. The approximation was chosen to be a polynomial because this framework is very flexible and can approximate any continuous function with a high enough polynomial order. The algorithm for determining the polynomial coefficients from a given altitude-dependent acceleration-limit profile is presented in \secref{Finding the Parameters}.

The bounded optimal controller structures in \eqsref{Bounded Optimal Controllers - zero and 1st} are comprised of intervals in which the constraint is active (the missile is saturated), and from intervals in which the constraint is inactive (the missile is unsaturated). To solve the problem, we need to find the times at which the constraint transitions from active to inactive, and vice versa. We call these times the ``switching times''. Since in guidance problems it is often more convenient to use the time-to-go instead of time, we find the “switching times-to-go” as part of the guidance law derivation.

\subsection{The Switching Times-to-Go}
We define the switching times-to-go as the intersection points between the absolute value of the unsaturated optimal acceleration command  (i.e., $|\alpha \bar{Z}(t_{f}) \varphi|$) and its time-to-go-dependent limit.
Qualitative acceleration command profiles for zero and first-order, strictly proper missile dynamics and a time-to-go polynomial (in this case, parabolic) missile acceleration command limit can be seen in \figref{Qualitatively Acceleration Profile}. At the beginning of the scenario, the missile has the highest maneuver capability, and when it climbs in altitude, the density decreases, and the missile's maneuverability decreases as well.
For typical ground-to-air interception scenarios and zero or first-order, strictly proper missile dynamics, the missile acceleration command can be divided into three intervals, which can be seen in \figref{Qualitatively Acceleration Profile}: before the saturation (where $t_{go} \in [ t_{{go}_{s_{2}}},t_{f} ]$), during the saturation (where  $t_{go} \in ( t_{{go}_{s_{1}}},t_{{go}_{s_{2}}} )$), and after the saturation (where $t_{go} \in [ 0,t_{{go}_{s_{1}}} ]$). A new function, flag, shall be defined by (presented in \figref{Qualitatively Acceleration Profile}
for the first-order dynamics)
\begin{equation} \label{Flag}
    \text{flag}(\xi) =
    \begin{cases}
			3, & \xi \in \left[t_{go_{s_{2}}},t_{f}\right]\\
            2, & \xi \in \textcolor{black}{\left(t_{go_{s_{1}}},t_{go_{s_{2}}}\right)}\\
            1, & \xi \in \textcolor{black}{\left[0,t_{go_{s_{1}}}\right]}
	\end{cases}
\end{equation}
\begin{figure}[hbt!]
\centering
\includegraphics[width=0.57\textwidth]{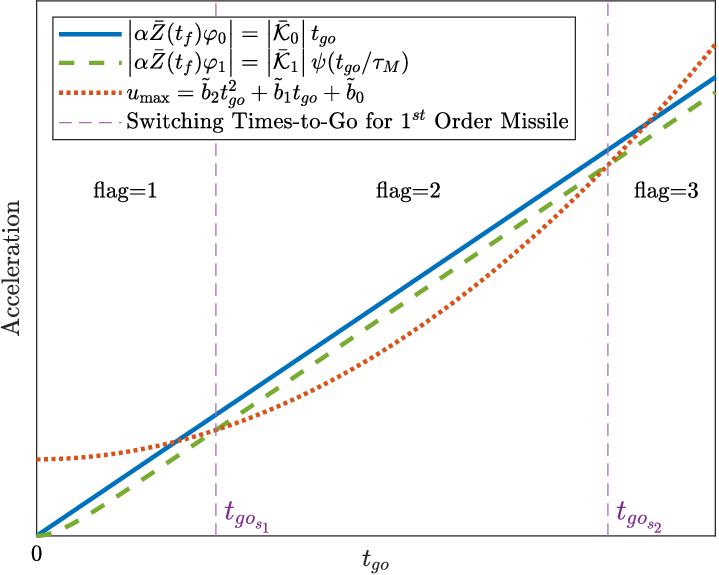}
\caption{Qualitative representation of $t_{go_{s_{1}}}$ and $t_{go_{s_{2}}}$ and the flag function.}
\label{Qualitatively Acceleration Profile}
\end{figure}

The derivation assumes a maximum of two intersection points between the absolute value of the unsaturated optimal acceleration command and its limit $u_{\max}$.
Equivalently, we assume the function $\ell$, defined by
\begin{equation} \label{l_def}
    \ell(\tgo)\triangleq \umax(\tgo)-\alpha |\bar{Z}(t_{f}) \varphi (\tgo) |
\end{equation}
has no more than two roots in $\bar{I}\triangleq \left[0,\tgo^{initial}\right]$.
Sufficient conditions to satisfy this assumption are: 
$\ell(\tgo)$, $\ell'(\tgo)$, and  $\ell''(\tgo)$ are continuous  in $\bar{I}$, differentiable in $I\triangleq \left(0,\tgo^{initial}\right)$, and,
\begin{subequations} \label{3cond}
\begin{align}
&u_{\max}'''(\tgo)\ge0, \quad \forall \tgo\in I \label{umax_cond}\\
&\abs{\varphi (\tgo)}''' \le0, \quad \forall \tgo\in I \label{uopt_cond} \\
&{u}_{\max}(\tgo=0)>\alpha \abs{\bar{Z}(t_{f})} \cdot \abs{\varphi (\tgo=0)} \label{zero_cond}
\end{align}
\end{subequations}
where $\Box '''$ denotes a third derivative with respect to the time-to-go. To improve readability, the proof that \eqsref{3cond} are sufficient conditions for having at most two intersection points is presented in \ref{Appendix_Proof}.
As mentioned in the introduction, the missile's altitude-dependent acceleration command limit is typically exponentially decreasing, i.e.,
\begin{equation} \label{u_max_A_beta}
    \umax(Y)=A{\mathrm e}^{-\beta Y}, \quad A,\beta>0
\end{equation}
where $Y$ is the missile's altitude. Under the assumption of small deviations from a collision course, and assuming the missile is launched from the ground, the missile's acceleration command limit is approximately
\begin{equation} \label{umax_approx_t}
    \umax(t)\approx A{\mathrm e}^{-\beta v_{M}\sin(\gamma_{M})t}
\end{equation}
Substituting \eqref{tgo} into \eqref{umax_approx_t}, we obtain
\begin{equation} \label{umax_tgo_app}
    \umax(\tgo)\approx C\cdot A{\mathrm e}^{\beta v_{M}\sin(\gamma_{M})\tgo}, \quad C={\mathrm e}^{-\beta v_{M}\sin(\gamma_{M})t_{f}}
\end{equation}
It is evident from \eqref{umax_tgo_app} that under these assumptions, the acceleration command limit is approximately exponentially increasing in $t_{go}$.
Due to the exponential nature of the acceleration limit, and because the polynomial $u_{\max}(t_{go})$ is an approximation of this exponential, the condition in \eqref{umax_cond} is likely to be satisfied for any polynomial order of $u_{\max}(t_{go})$.  Note that for a parabolic acceleration command limit, this condition is always satisfied. Since the designer chooses the polynomial approximation order, the condition in \eqref{umax_cond} can always be satisfied (an extended discussion is presented in \secref{Finding the Parameters}).

Moreover, from \eqsref{phi0_phi1} we obtain $\varphi_{0}(\tgo=0)=\varphi_{1}(\theta=0)=0$. Since the missile always has some maneuverability, particularly at $\tgo=0$, the condition in \eqref{zero_cond} is satisfied for both zero-order and first-order, strictly proper missile dynamics. Additionally, since
\begin{equation}
\abs{\varphi_{0} (\tgo)}'''\equiv0, \quad \abs{\varphi_{1} (\tgo)}'''=
        \frac{1}{\tau_{M}^{3}}\frac{\partial^{3} \abs{\varphi_{1} (\theta)}}{\partial \theta^{3}}=-\frac{\abs{K_{M}}}{\tau_{M}^{2}}{\mathrm e}^{-\tgo/\tau_{M}}   
\end{equation}
The condition in \eqref{uopt_cond} is satisfied for both zero-order and first-order, strictly proper missile dynamics as well. Furthermore, the smoothness requirements on $\ell(\tgo)$ are also met for the polynomial acceleration limits and the aforementioned missile dynamics.  

In conclusion, for linear, zero and first-order, strictly proper missile dynamics, \TWOeqsref{uopt_cond}{zero_cond} are guaranteed to be satisfied, \eqref{umax_cond} is excepted to be satisfied, and if for a given polynomial it is not satisfied, one should change the polynomial order, or choose $u_{\max}(\tgo)$ to be parabolic, so the condition is satisfied automatically. Thus, a maximum of two intersection points between the optimal unsaturated acceleration command curve and the acceleration bounds curve is guaranteed.
If there are two intersection points, the scenario starts unsaturated at flag=3, moves on to the saturated region at which flag=2, and ends at the last unsaturated part at which flag=1.
If there is only one intersection point, the scenario starts in the saturated region (flag=2) and ends in the last unsaturated region (flag=1). There is no possibility that in the case of one intersection point the scenario starts at the unsaturated region (flag=3), continues to the saturated  region, and stays there until the end of the scenario, since
\begin{equation} \label{lim_u}
     \lim_{t_{go}\to0} \left|\bar{u}^{*}_{0}(t_{go})\right|=\lim_{\theta\to0} \left|\bar{u}^{*}_{1}(\theta)\right|=  0 < \tilde{b}_{0} =  \lim_{t_{go}\to0} u_{\max}(t_{go})
\end{equation}
\eqref{lim_u} implies that there is always an interval at the end of the scenario in which the missile is unsaturated, as long as the free coefficient of the polynomial $u_{\max}(\tgo)$, i.e. $\tilde{b}_{0}$, is greater than zero.
As mentioned before, we assume the missile always has some maneuverability, particularly at $\tgo=0$, and therefore, $\tilde{b}_{0}>0$. However, it should be noted that if the terminal miss distance, $\bar{Z}(t_f)$, is large, the unsaturated interval at the end of the scenario might be very short.
If there are no intersection points, the entire scenario is unsaturated, i.e., flag=1.

\subsection{Calculating the Terminal ZEM} \label{Ztf_BOGL}
It is evident from \eqref{Bounded Optimal controller} that $\bar{Z}(t_{f})$ must be computed to find the optimal controller. We compute it by integration of \eqref{Z_dot} from $t$ to $t_{f}$. Since the calculation depends on the value of $\text{flag}(t_{go})$, let us calculate $\bar{Z}(t_{f})$ for every possible value of $\text{flag}(t_{go})$. During the derivation, we assume the switching times-to-go are known. The algorithm to find them is discussed in \secref{Algo_SwitchingTimes}.

\subsubsection{\texorpdfstring{$\text{flag}(t_{go})=3$ (the missile is before the saturation region)}{The missile is before the saturation region}} 
Substituting \TWOeqsref{Derivative}{Bounded Optimal controller} into \eqref{Z_dot} and integrating, recalling saturation occurs between $t_{go_{s_{1}}}$ and $t_{go_{s_{2}}}$, yields
 \begin{equation} \label{3}
 \begin{split}
    &\bar{Z}(t_{f}) = Z(t)- \alpha \bar{Z}(t_{f})\int_{0}^{t_{go_{s_{1}}}} \varphi^{2}(\xi) \, \dd\xi +   \text{sign}\left[\bar{Z}(t_{f})\right] \int_{t_{go_{s_{1}}}}^{t_{go_{s_{2}}}} \varphi(\xi) u_{\max}(\xi) \, \dd\xi - \alpha \bar{Z}(t_{f})\int_{t_{go_{s_{2}}}}^{t_{go}} \varphi^{2}(\xi) \, \dd\xi \\
\end{split}
\end{equation}
Additionally, the following statement always holds (see \ref{Appendix_Derivation} for proof): $\text{sign}\left[ \bar{Z}(t_{f})\right]=\text{sign}\left[ Z(t)\right]$. Rearranging \eqref{3}  
we obtain an expression for $\bar{Z}(t_{f})$ as a function of $t_{go}$ , $t_{go_{s_{1}}}$, and $t_{go_{s_{2}}}$
\begin{equation} \label{Ztf_3}
    \bar{Z}(t_{f})\big|_{\text{flag}(t_{go})=3} = \frac{Z(t)+ \text{sign}\left[Z(t)\right] \left[ I_{1}(t_{go_{s_{2}}})-I_{1}(t_{go_{s_{1}}}) \right]}{1+\alpha \left[I_{2}(t_{go})-I_{2}(t_{go_{s_{2}}})+I_{2}(t_{go_{s_{1}}}) \right]}
\end{equation}
where
\begin{equation} \label{integrals}
    I_{1}(t_{go})= \int_{0}^{t_{go}} u_{\max}(\xi) \varphi(\xi) \, \dd\xi, \quad I_{2}(t_{go})=  \int_{0}^{t_{go}} \varphi^{2}(\xi) \, \dd\xi
\end{equation}

\subsubsection{\texorpdfstring{$\text{flag}(t_{go})=2$ (the missile is in the saturation region)}{The missile is in the saturation region}}
Similar to the previous case, keeping in mind that saturation occurs between $t_{go_{s_{1}}}$ and $t_{go}$, we obtain 
 \begin{equation} \label{2}
    \bar{Z}(t_{f}) = Z(t)- \alpha \bar{Z}(t_{f}) \int_{0}^{t_{go_{s_{1}}}} \varphi^{2}(\xi) \, \dd\xi  + \text{sign}\left[\bar{Z}(t_{f})\right] \int_{t_{go_{s_{1}}}}^{t_{go}} \varphi(\xi) u_{\max}(\xi) \, \dd\xi 
\end{equation}
Substituting \eqref{integrals} into \eqref{2} and rearranging, we obtain an expression for $\bar{Z}(t_{f})$ as a function of $t_{go}$ and $t_{go_{s_{1}}}$

 \begin{equation} \label{Ztf_2}
    \bar{Z}(t_{f}) \big|_{\text{flag}(t_{go})=2} = \frac{Z(t)+ \text{sign}\left[Z(t)\right] \left[ I_{1}(t_{go})-I_{1}(t_{go_{s_{1}}}) \right]}{1+\alpha I_{2}(t_{go_{s_{1}}})}
\end{equation}

\subsubsection{\texorpdfstring{$\text{flag}(t_{go})=1$ (the missile is after the saturation region)}{The missile is after the saturation region}}
Similarly to the previous cases, keeping in mind that there is no saturation in this case, we obtain
 \begin{equation} \label{Corflag1}
    \bar{Z}(t_{f}) = Z(t) - \alpha \bar{Z}(t_{f})\int_{0}^{t_{go}} \varphi^{2}(\xi) \, \dd\xi 
\end{equation}
Substituting \eqref{integrals} into \eqref{Corflag1} and rearranging, we obtain an expression for $\bar{Z}(t_{f})$ as a function of $t_{go}$ 
\begin{equation} \label{Ztf_1}
\bar{Z}(t_{f}) \big|_{\text{flag}(t_{go})=1}=\frac{Z(t)}{1+\alpha I_{2}(t_{go})}
\end{equation}
which is identical to the $Z(t_{f})$ expression of the well-known unbounded OGL.

It is clear from \THREEeqsref{Ztf_3}{Ztf_2}{Ztf_1} that $Z(t)$ has to be calculated to compute the optimal controller. 
Substituting \eqref{DPhi} into \eqref{Z} and rearranging yields  
\begin{equation} \label{eq:Z_solution_general}
    Z(t) = y + \dot{y} t_{go}+ {\bm{\phi}}_{{1}{M}}^{T}\mathbf{x}_{M} + Z_{T}(t), \quad
    Z_{T}(t) = K_T\mathscr{L}_{t_{go}}^{-1}\left[\tilde{a}_{T}(s)/ s^{2} \right]
\end{equation}
where $\tilde{a}_{T}(s)$ is the target acceleration in the Laplace domain, which is given by
\begin{equation} \label{eq:a_T}
    \tilde{a}_{T}(s) =   \left[\mathbf{c}_{T}^{T}\left({s}\mathbf{I}-\mathbf{A}_{T}\right)^{-{1}}\right]\mathbf{x}_{T}+G_T(s)\tilde{w}, \quad 
    G_{T}(s) = \mathbf{c}_T^T\left({s}\mathbf{I}-\mathbf{A}_T\right)^{-{1}}\mathbf{b}_T+d_T
\end{equation}
$G_{T}(s)$ is the target acceleration transfer function and $\tilde{w}(s)$ is the Laplace transform of the target acceleration command. 
Note that the first term of $\tilde{a}_{T}(s)$ is the response to the initial conditions, and the second is the response to the acceleration command. 
To improve readability, the full derivation of \eqref{eq:Z_solution_general} is presented in \ref{Appendix_Derivation}.
\begin{rem}
    From \eqref{eq:Z_solution_general} it is clear that, unlike we originally assumed, we do not need to know the target acceleration command, $w$, from the current time onward, and it is enough to know the target acceleration, $a_{T}$, from the current time onward, which is a typical assumption in optimal-control-based guidance laws.  
\end{rem}
For a constant target acceleration, which is the case considered in this paper, we obtain
\begin{equation} \label{step}
    a_{T}(t_{go})=\bar{a}_{T}\cdot \mathbbm{1}(t_{go}) \quad \Longrightarrow \quad 
    \tilde{a}_{T}(s) = \mathscr{L}_{t_{go}}\left[a_{T}(t_{go})\right]=\frac{\bar{a}_{T}}{s}
\end{equation}
where $\mathbbm{1}$ is the step function and $\bar{a}_{T}$ is the constant target acceleration value. Substituting \eqref{step} into \eqref{eq:Z_solution_general} we obtain
\begin{equation} \label{4}
Z_{T}(t)=
K_T\mathscr{L}_{t_{go}}^{-1}\left(\frac{\bar{a}_{T}}{s^3}\right)=\frac{K_T}{2} \bar{a}_{T} t_{go}^2
\end{equation}
Substituting \eqref{4} into \eqref{eq:Z_solution_general}, we obtain $Z(t)$ for arbitrary, zero, and first-order, strictly proper linear missile dynamics, respectively (see \ref{Appendix_Derivation})
 \begin{subequations} \label{Z_noM_0_1st}
 \begin{align} \label{Z_noM_0}
 &Z\left(t\right)=y+\dot{y}t_{go}+\frac{K_T}{2} \bar{a}_{T} t_{go}^2+\bm{\phi}_{{1}M}^T\mathbf{x}_M \\
 &Z_0\left(t\right)=y+\dot{y}t_{go}+\frac{K_T}{2} \bar{a}_{T} t_{go}^2 \\ \label{Z_noM_1}
 &Z_1\left(t\right)=y+\dot{y}t_{go}+\frac{K_T}{2}{\bar{a}}_T t_{go}^2-K_M\tau_M^2\psi(\theta)a_{M}
 \end{align} 
 \end{subequations}

\subsection{Difference Between BOGL and OGL Under Constant and Time-Varying Bounds} \label{proof_BOGL}
Refs. \cite{rusnak1990optimal, rusnak1991optimal} showed that, for constant hard acceleration command bounds and minimum-phase missile dynamics, the classical OGL remains optimal even in the presence of saturation; namely, the BOGL and the saturated OGL yield identical closed-loop behavior.
In the constant-bound case, saturation occurs from the beginning of the engagement, leaving no unsaturated interval in which the controller can compensate in advance for the future loss of maneuverability. Consequently, once the trajectory exits saturation, the BOGL and the saturated OGL evolve identically. This behavior follows from the fact that, for minimum-phase missile dynamics, the optimal unsaturated command is monotonically increasing in $t_{go}$, resulting in a single intersection with the constant acceleration command limit.

Importantly, this property does not generally hold when the acceleration command bound is time-varying. In such cases, multiple transitions between unsaturated and saturated arcs may occur (e.g., the unsaturated–saturated–unsaturated profile corresponding to flag=3 in \figref{Qualitatively Acceleration Profile}). The existence of an unsaturated interval prior to saturation enables the BOGL to anticipate the upcoming constraint activation and maneuver more aggressively before entering the saturated region. As a result, the BOGL can achieve substantially different behavior from the saturated OGL. 
However, similar to the constant-bound case, when no unsaturated interval exists prior to saturation (e.g., the scenario starts at flag=2 in \figref{Qualitatively Acceleration Profile}), the BOGL reduces to the saturated OGL.

Let us show that the BOGL generates larger acceleration command magnitudes than the OGL in the flag=3 region when the OGL initially remains unsaturated.
\begin{proof}
It follows from the optimal controller structures in \TWOeqsref{u_opt_free}{Bounded Optimal controller} that if the OGL is not saturated (i.e., flag=3), larger terminal ZEM values lead to larger acceleration commands. 
Therefore, we proceed to show that the absolute value of the terminal ZEM associated with the BOGL is greater than its OGL counterpart, i.e., $\abs{\bar{Z}(t_{f})}>\abs{Z(t_{f})}$.
The OGL terminal ZEM, $Z(t_{f})$, follows directly from \eqref{Ztf_1}, since the OGL derivation neglects saturation effects. The BOGL terminal ZEM, 
$\bar{Z}(t_{f})$, for the flag=3 case, is obtained using \eqref{3}.
Since saturation occurs in $(t_{go_{s_{1}}},t_{go_{s_{2}}})$, the absolute value of the optimal command is smaller than its unsaturated counterpart within this interval. This follows directly from the definition of the saturated region. Therefore, the following relation holds
\begin{equation} \label{eq_for_proof}
\abs{\int_{t_{go_{s_{1}}}}^{t_{go_{s_{2}}}} \varphi(\xi) u_{\max}(\xi) \, \dd\xi} < \abs{\alpha \bar{Z}(t_{f})\int_{t_{go_{s_{1}}}}^{t_{go_{s_{2}}}} \varphi^{2}(\xi) \, \dd\xi}    
\end{equation}
Without loss of generality, let us assume that $Z(t)>0$. Substituting \TWOeqsref{eq_for_proof}{integrals} into \eqref{3}, noting that $\text{sign}\left[ \bar{Z}(t_{f})\right]=\text{sign}\left[ Z(t)\right]$, and that for minimum-phase missiles $\varphi(\xi)\le 0, \forall \xi\ge 0$, we obtain
\begin{equation} \label{eq2_for_proof}
    \bar{Z}(t_{f}) > Z(t) - \alpha \bar{Z}(t_{f})\int_{0}^{t_{go}} \varphi^{2}(\xi) \, \dd\xi = Z(t) - \alpha \bar{Z}(t_{f})I_{2}(t_{go})
\end{equation}
Rearranging \eqref{eq2_for_proof} and substituting \eqref{Ztf_1} (with $Z(t_{f})$ instead of $\bar{Z}(t_{f})$ for OGL), we obtain
\begin{equation} \label{proof}
    \bar{Z}(t_{f})>\frac{Z(t)}{1+\alpha I_{2}(t_{go})}=Z(t_{f})
\end{equation}
which, through the optimal controller structures in \TWOeqsref{u_opt_free}{Bounded Optimal controller}, proves that the BOGL acceleration command magnitude exceeds that of the OGL in the initial unsaturated interval (i.e., flag=3).
\end{proof}

\subsection{The Existence of a Solution for Target Interception}\label{Sec:SolutionExistence}
The maximum correction the missile can perform from the current moment to the end of the scenario is achieved by commanding its acceleration command limit throughout the scenario. Although it is not the optimal strategy (since the optimal controller approached zero as $t_{go}\to 0$ and $\umax(t_{go})$ does not, as \eqref{lim_u} shows),
this case is worth discussing to derive a necessary condition for the existence of a solution to the problem with a small miss distance.
Substituting the acceleration limit as the acceleration command into \eqref{Z_dot}, integrating, and using the definitions in \eqref{integrals} yields
\begin{equation} \label{Existence}      
\bar{Z}(t_{f})=Z(t)+\sign{\left[Z(t)\right]}I_{1}(\tgo)
\end{equation}
That is, a necessary condition for the existence of a solution with a small miss distance is 
\begin{equation} \label{Existence_criteria}  
\abs{I_{1}(\tgo)}\ge\abs{Z(t)}\textcolor{black}{-m}
\end{equation}
where $m$ is the maximal miss allowed.
Otherwise, even a controller that stays saturated during the whole scenario, utilizing the missile's maneuver capability to its maximum, will not be sufficient, and the missile will miss the target.

\section{Algorithm Implementation} \label{Algorithm Implementation}

\subsection{Approximation of the Acceleration Command Bounds} \label{Finding the Parameters}
To implement the BOGLs, the coefficients $\{\tilde{b}_{i}\}_{0}^{n}$ such that $u_{\max}(t_{go})$ in \eqref{umax(tgo)}
is a good approximation of the real, altitude-dependent, absolute value of the missile's acceleration limit need to be found.
To this end, an unsaturated OGL (\eqsref{Unbounded Optimal Controllers - zero and First}) simulation is performed, and the real acceleration command limit is calculated as a function of the approximated time-to-go. That is, the OGL-guided missile's altitude is evaluated to assign the acceleration limit as a function of the approximated time-to-go. We then use the weighted least squares algorithm to approximate $u_{\max}(t_{go})$. The polynomial order is determined by the designer, as long as the condition in \eqref{umax_cond} is satisfied. For a third-order polynomial, this condition reduces to simply $\tilde{b}_{3}>0$, and for a fourth-order polynomial it reduces to $\tilde{b}_{3}>0$ and  $\tilde{b}_{4}>-\tilde{b}_{3} / (4 \tgo^{initial})$. From our experience, a fourth-order polynomial is sufficient to describe the acceleration command limit, and we only needed a third-order polynomial in our simulations in \secref{Simulations}.
As mentioned before, if the condition is not satisfied, one can choose a second-order polynomial, so \eqref{umax_cond} is satisfied automatically, and repeat the parameters estimation process. When choosing a parabolic polynomial approximation, $\tilde{b}_{2}$ must be positive, so the acceleration limit approximation captures the convex behavior of the real, exponential, altitude-dependent limits.
Since the accuracy of a second-order polynomial approximation might be insufficient, and since the most crucial portion of the approximation is in the saturated region, in the estimation process the samples in the saturation region (i.e., flag=2) can be weighted more than the other samples. The identification of these samples is also based on the initial unsaturated OGL simulation (i.e., the intersections between the OGL command and the actual acceleration limit). 
\begin{rem}
As shown in \eqref{umax_tgo_app}, $u_{max}(t_{go})$ is approximately exponential. Therefore, it is expected that its polynomial approximation would uphold the condition in \eqref{umax_cond}. This is also what we observed in all the simulations we performed in this study using third and fourth-order polynomials. The second-order option is provided as a backup.
\end{rem}

The estimation model for the coefficients $\{\tilde{b}_{i}\}_{0}^{n}$ is the classical linear model 
\begin{equation} 
    z(k) = \mathbf{h}^{T}(k)\mathbf{\tilde{b}}+\textcolor{black}{\nu}(k), \quad \mathbf{h}(k) = \left[ 1 \quad \tgo(k) \quad ... \quad \tgo^{n}(k) \right]^{T}, \quad \mathbf{\tilde{b}} = \left[ \tilde{b}_{0} \quad \tilde{b}_{1} \quad ... \quad \tilde{b}_{n} \right]^{T}
\end{equation}
where $\tgo(k)\triangleq t_{f}-t_{k}$ is the approximated time-to-go in discrete time, and $z(k),\textcolor{black}{\nu}(k)$ are the acceleration limit and the noise at $\tgo=\tgo(k)$, respectively. The weighted least squares estimator of $\mathbf{\tilde{b}}$ is \cite{mendel1995lessons}
\begin{equation}   \label{WLS_pinv} 
\mathbf{\hat{\tilde{b}}}=\left(\mathbf{Q}^{\frac{1}{2}} \mathbf{H}\right)^{\dagger}\mathbf{Q}^{\frac{1}{2}}\mathbf{z}, \quad \mathbf{\textcolor{black}{\Gamma}}^{\dagger}\triangleq\left(\mathbf{\textcolor{black}{\Gamma}}^T\mathbf{\textcolor{black}{\Gamma}}\right)^{-1}\mathbf{\textcolor{black}{\Gamma}}^T
\end{equation}
where $\mathbf{Q}^{\frac{1}{2}}\mathbf{Q}^{\frac{1}{2}}=\mathbf{Q}$, $\mathbf{\textcolor{black}{\Gamma}}^{\dagger}$ is the Moore-Penrose Pseudoinverse of $\mathbf{\textcolor{black}{\Gamma}}$, and
\begin{equation}
 \mathbf{z} = \left[ z(1) \quad z(2) \quad ... \quad z(N) \right]^{T}, \quad \mathbf{H} = \left[ \mathbf{h}(1) \quad \mathbf{h}(2) \quad ... \quad \mathbf{h}(N) \right]^{T}, \quad \mathbf{Q}=\text{diag}\left\{\mathbf{q}^{T}_{1} ,\mathbf{q}^{T}_{2},\mathbf{q}^{T}_{3}\right\}    
\end{equation}
where $N$ is the number of measurements, $\mathbf{Q}$ is the weight matrix, and $\mathbf{q}_{i}$ are identity vectors of appropriate order multiplied by a tuning parameter $q_{i}$, for each one of the regions of the $\text{flag}(t_{go})$ function. 
\subsection{The Switching Times-to-Go and the Terminal ZEM} \label{Algo_SwitchingTimes}
We obtained three different expressions for $\bar{Z}(t_{f})$, one for each possible value of $\text{flag}(t_{go})$. The value of $\bar{Z}(t_{f})$ is then substituted into the optimal controller structure in \eqsref{Bounded Optimal Controllers - zero and 1st}. However, to calculate $\bar{Z}(t_{f})$, the values of $\text{flag}(t_{go})$ and the switching times-to-go must be computed.
In this section, we propose a general algorithm for finding the switching times-to-go and $\bar{Z}(t_{f})$. It can be applied to any order acceleration command limit polynomial and both dynamic models. For brevity, it is only presented here for the first-order, strictly proper missile dynamics. As will be presented further on, the proposed algorithm requires iterations to converge, and is numerical in nature. To initialize the iterative process, we use results from a zero-order missile and a parabolic acceleration command limit. This combination offers an analytical solution for the switching times-to-go as a function of $\bar{Z}(t_{f})$, simplifying the algorithm and improving convergence time. For better readability, the initialization algorithm is presented in \ref{Appendix_Algo}.

Normalizing \eqref{umax(tgo)} and substituting it together with the unsaturated optimal controller into \eqref{l_def} yields
\begin{equation}\label{l}
\mathcal{\ell} \left( \theta \right) = 
\sum_{i=0}^{n} b_{i} \theta^{i} - \left|\bar{\mathcal{K}}_{1}\right|\psi \left( \theta \right), \quad b_{i} = \tilde{b}_{i} \tau_{M}^{i}, \quad \forall i\in\left\{0,1,...,n\right\}, \quad \theta \dfn \frac{t_{go}}{\tau_{M}}
\end{equation}
The roots of $\ell(\theta)$ are the intersection points between the unsaturated optimal acceleration command and its normalized-time-to-go dependent limit, i.e., the normalized switching times-to-go.
Assuming $\bar{Z}(t_{f})$ is known (recalling the definition of $\bar{\mathcal{K}}_{1}$ in \eqref{Bounded Optimal Controller - 1st}), these roots can be computed.
Since $\psi(\theta)$ is a transcendental function, analytical solutions cannot be found, and the normalized switching times-to-go are found numerically.
The algorithm for finding the normalized switching times-to-go and $\bar{Z}(t_{f})$ is an iterative process, in which we initially guess the value of $\bar{Z}(t_{f})$ to calculate the roots of $\ell(\theta)$, and then use them to recompute the new value of $\bar{Z}(t_{f})$, to convergence. 
For the first time step, the initial guess for $\bar{Z}(t_{f})$ is the value resulting from the zero-order, parabolic acceleration limits (\algoref{Ztf0_algo} in \ref{Appendix_Algo}), and for the other time steps, the initial guess is the previous time step's value of $\bar{Z}(t_{f})$. The algorithm is presented in \algoref{Ztf1_algo}, and 
uses the Newton-Raphson  \cite{suli2003introduction} method with deflation 
to find all the roots of $\ell(\theta)$. The update equation of the Newton-Raphson algorithm is given by 
\begin{equation} \label{Newton-Raphson_eqn_l}
    \theta_{k+1} = \theta_{k} - \frac{\mathcal{\ell}(\theta_{k})}{\mathcal{\ell}'(\theta_{k})}, \quad 
    \mathcal{\ell}'(\theta)=  \sum_{i=1}^{n} i b_{i} \theta^{i-1} + \left|\bar{\mathcal{K}}_{1}\right| \left({\mathrm e}^{-\theta}-1 \right) 
\end{equation}
and the full algorithm of finding all the roots with deflation is presented in \algoref{NR_deflation}.
We denote $\text{NR}(H,H')$ as the result of the Newton-Raphson algorithm for the function $H$ with a derivative of $H'$, where the term $\bar{Z}(t_{f})$ in $\ell(\theta)$ and $\ell'(\theta)$ is the latest one calculated, and $\varepsilon$ is a small number.
The terms $\Delta \theta_{s_{1}},\Delta \theta_{s_{2}}$ are the differences between the first and the second normalized switching times-to-go between two consecutive iterations, respectively, and $\theta_{tol}$ is the convergence termination criteria on $\Delta \theta_{s_{1}},\Delta \theta_{s_{2}}$ between iterations.

\begin{algorithm}[hbt!]
\caption{Calculating $\bar{Z}(t_{f})$ for First-Order Missile Dynamics with an Arbitrary-Order Polynomial Acceleration Limit.}
\label{Ztf1_algo}
    \begin{algorithmic}[1]
        \State Compute $\bar{Z}(t_{f})$ assuming flag=1 (\eqref{Ztf_1}).
        \State Compute the roots of $\ell(\theta)$ with the $\bar{Z}(t_{f})$ of flag=1, $\Theta=\text{NR\_all}\left[\ell(\theta)\right]$ (\algoref{NR_deflation}).       
        \If{No roots or $\Theta_{1}>\theta$}
            \State The missile is in the flag=1 region.
            \State Return $\bar{Z}(t_{f})$.
        \Else
        \If {first time step}
        \State \parbox[t]{\dimexpr\linewidth-1\algorithmicindent} {Use $\bar{Z}(t_{f})$ from \algoref{Ztf0_algo} to compute $\tgo{_{s_{1}}},\tgo{_{s_{2}}}$, and use them to compute $\bar{Z}(t_{f})$ as an initial guess according  to the flag value (\THREEeqsref{Ztf_3}{Ztf_2}{Ztf_1}).} 
        \Else
        \State Use $\bar{Z}(t_{f})$ from the last time step as an initial guess.
        \EndIf        
        \State Compute the roots of $\ell(\theta)$ with the last updated $\bar{Z}(t_{f})$, $\Theta=\text{NR\_all}\left[\ell(\theta)\right]$ (\algoref{NR_deflation}).       
        \If{No roots or $\Theta_{1}>\theta$}
        \State Compute $\bar{Z}(t_{f})$ for flag=1 (\eqref{Ztf_1}).
        \ElsIf{$\Theta_{1}<\theta<\Theta_{2}$}
        \State Compute $\bar{Z}(t_{f})$ for flag=2 with $\theta_{s_{1}}=\Theta_{1}$ (\eqref{Ztf_2}).
        \Else
        \State Compute $\bar{Z}(t_{f})$ for flag=3 with $\theta_{s_{1}}=\Theta_{1},\theta_{s_{2}}=\Theta_{2}$ (\eqref{Ztf_3}).
        \EndIf
        \State Recompute the normalized switching times-to-go and $\bar{Z}(t_{f})$ according to lines 12-19.
        \While {$\left|\Delta \theta_{s_{1}} \right|,\left|\Delta \theta_{s_{2}} \right|>\theta_{tol}$}
        \State Recompute the normalized switching times-to-go and $\bar{Z}(t_{f})$ according to lines 12-19.
        \EndWhile
        \State Return $\bar{Z}(t_{f})$.
    \EndIf
    \end{algorithmic}
\end{algorithm}

\begin{algorithm}[hbt!]
\caption{Newton-Raphson for Finding All Roots of $H(\theta)$ with Deflation.}\label{NR_deflation}
    \begin{algorithmic}[1]
        \State Compute a root using $\Theta_{1}=\text{NR}\left[H_{1}(\theta),H'_{1}(\theta)\right]$.
        \While {Latest $\Theta_i$ is not NaN}
            \State Define $H_{i+1}(\theta)=\frac{H_{i}(\theta)}{\theta-\Theta_{i}}$.
            \State Define $H'_{i+1}(\theta)=\frac{H'_{i}(\theta)\left(\theta-\Theta_{i}\right)-H_{i}(\theta)}{\left( \theta-\Theta_{i}\right)^{2}}$.
            \State Compute a new root using $\Theta_{i+1}=\text{NR}\left[H_{i+1}(\theta),H'_{i+1}(\theta)\right]$.           
        \EndWhile
        \State Remove the NaN value, and then sort the set $\Theta$ from the minimal value to the maximal value.
        \State Return all $\Theta_{i}$ for which $\sign\left[{H_{1}(\Theta_{i}+\varepsilon})\right]\neq\sign\left[{H_{1}(\Theta_{i}-\varepsilon})\right]$
    \end{algorithmic}
\end{algorithm}
\begin{rem}
An alternative numerical algorithm for calculating the roots of $\ell(\theta)$ that was also considered is based on approximating the exponential term in $\ell(\theta)$. For relatively ``large'' normalized switching times-to-go (e.g., $\theta>5$) , the exponential term in $\ell(\theta)$ is very small. Neglecting this term, $\ell(\theta)$ becomes a polynomial whose roots are much simpler to find, compared to the Newton-Rhapson-based root-finding method discussed in this section. They can be calculated by a non-iterative algorithm that finds them all (like the Matlab \textit{roots} function, which computes the eigenvalues of the corresponding companion matrix). For ``small'' normalized switching times-to-go, the exponential terms can be approximated by a Taylor series, and $\ell(\theta)$ becomes a polynomial again.\\ The results presented in \secref{Simulations} were generated with the full algorithms presented in this section. However, the alternative approach was also evaluated, and very minor differences were obtained when using the approximations discussed in this remark.
\end{rem}

\subsection{Approximation of the ZEM} 
As presented before, $Z(t)$ in \eqsref{Z_noM_0_1st} must be computed to obtain the optimal controller. To this end, we need to approximate $y+\dot{y}t_{go}$ from the measurements available to the missile in a nonlinear setting (simulation or real-world measurements). Assuming small deviations from the collision triangle, we obtain
\begin{equation} \label{phi_approx}
\sigma-\sigma(0)=\arcsin{\left(\frac{y}{r}\right)}\approx\frac{y}{r}, \quad  \Longrightarrow \quad \dot{\sigma}=\frac{1}{v_ct_{go}^2}\left(y+\dot{y}t_{go}\right) \quad
    \Longrightarrow \quad y+\dot{y}t_{go}=\dot{\sigma}v_ct_{go}^2
\end{equation}
This implementation is used in all guidance laws simulations.

\section{Simulations} \label{Simulations}
This section presents nonlinear simulations demonstrating the proposed guidance laws. The first subsection compares the first-order BOGL with its corresponding OGL (i.e., the same dynamic model without explicitly accounting for saturation constraints in the derivation), and discusses the resulting differences in behavior and performance.
The second subsection compares the zero-order BOGL with a modified version of the guidance law derived in \cite{ben2003new}. To ensure a fair comparison, both guidance laws are implemented using the same zero-order missile dynamic model. The modified guidance law employs an exponential time-to-go weighting of the control-effort running cost and is referred to in this paper as ``New Augmented Proportional Navigation'' (NAPN), rather than ``New Proportional Navigation'' as in \cite{ben2003new}, because future target acceleration was additionally incorporated, to enable a fair comparison with the proposed BOGL. In all simulations, the missile and the target fly at constant speeds, and the weighting parameter is $\alpha=10^{6} \ (1/s^3)$ (for all guidance laws: BOGL, OGL, NAPN). In addition, to simulate a more realistic scenario, a blind range of $100 (m)$ was used. Within the blind range, the acceleration command is set to the last one calculated before entering the blind range.
The ODEs were integrated using the MATLAB \textit{ode45} function, with a controller frequency of 100 Hz.
In all simulations, the missile's acceleration command bound is a realistic, altitude-dependent limitation. 
The altitude-dependent missile's acceleration command limit for all guidance laws is similar to the one in \cite{taub2013intercept}, and is given by \eqref{u_max_A_beta}, with $A=250 \ (m/s^2)$ and $\beta=9.2\times10^{-5} \ (1/m)$.

\subsection{The First-Order BOGL Performance and Comparison to the OGL}
To simulate a more realistic scenario, linear, first-order, strictly proper dynamics are used for both the missile and the target. Therefore, the first-order BOGL and OGL are applied to the missile.
The scenario parameters are shown in \tableref{Parameters}, and the target acceleration is $-5g$ (i.e., downward).
\begin{table}[hbt!]
\caption{\label{Parameters} Scenarios parameters.}
\centering
\begin{tabular}{lccccc}
\hline\hline
& $v\left( \frac{m}{s}\right)$ & $\tau(s)$ & $\gamma(0)(deg)$ & $x(0)(km)$& $y(0)(km)$\\\hline
Missile& 1800& 0.2& 69& 0& 0\\
Target& 1500& 0.2& -20& 60& 45\\
\hline\hline
\end{tabular}
\end{table}
In the BOGL implementation, the acceleration limits are approximated as a third-order polynomial of the approximated time-to-go (\secref{Finding the Parameters}). The weighting parameters for the polynomial coefficients estimation are  $q_{1}=q_{2}=q_{3}=1$, i.e., the classical least squares estimator. The guidance law used to generate the trajectory for estimating the polynomial coefficients is a first-order OGL (\eqref{Unbounded Optimal Controller - 1st}).  

Figure \ref{Trajectories_SATalt_1stOrder} presents the trajectories of the BOGL and the OGL-guided missiles, along with the target trajectory.
It also presents both guidance laws' control efforts and miss distances. It is evident that the differences in control effort between the guidance laws are negligible. However, as can be seen in the zoomed-in view at the end of the scenario in the bottom-right corner, the
OGL misses the target by hundreds of meters, while the BOGL hits it perfectly.
It should be noted that the total cost function of the BOGL is orders of magnitude lower than that of the OGL, due to large differences in miss distance.

The reason for the miss distance difference can be seen in \figref{AccCommand_SATalt_1stOrder}. It is evident that the OGL is saturated from $t_{go}\approx 17.5(s)$ until the end of the scenario. The divergence of the yellow dotted line, which is the OGL's unsaturated acceleration command, towards the end of the scenario, is a testament to the lack of maneuverability of the OGL-guided missile, which led to a miss. On the other hand, the BOGL-guided missile accounts for the predicted saturation. It uses larger acceleration commands (in absolute value) before entering the saturation region, enabling it to exit the saturation region before the scenario ends and resulting in a negligible miss distance. It should be noted that the BOGL enters the saturation region earlier than the OGL, to exploit more of its capabilities.
Moreover, it is evident in \figref{AccCommand_SATalt_1stOrder} that the time-to-go dependent acceleration limits are very close to the real, altitude-dependent ones. In Addition, the BOGL's real acceleration limit is very close to the OGL's one, which justifies the use of an OGL simulation to estimate the BOGL time-to-go approximated acceleration limit.

\begin{figure}[hbt!]
\centering
\includegraphics[width=0.57\textwidth]{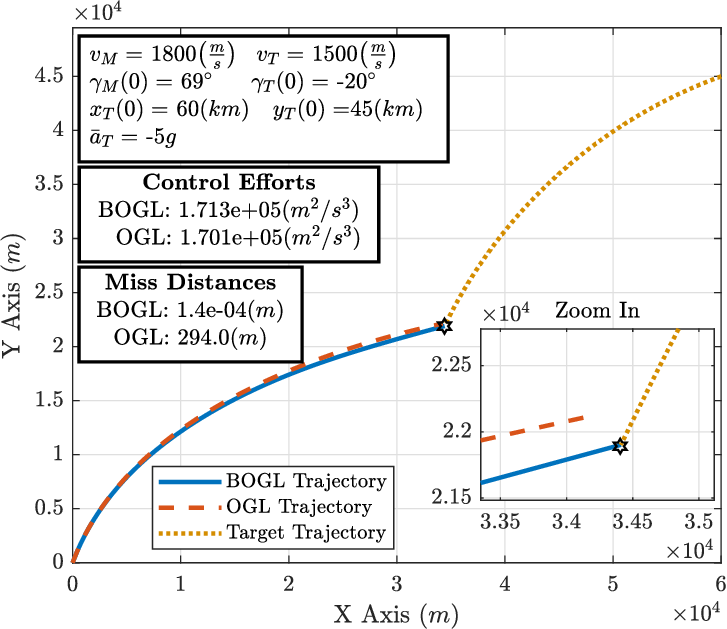}
\caption{\label{Trajectories_SATalt_1stOrder} Missiles and target trajectories.}
\end{figure}

\begin{figure}[hbt!]
\centering
\includegraphics[width=0.57\textwidth]{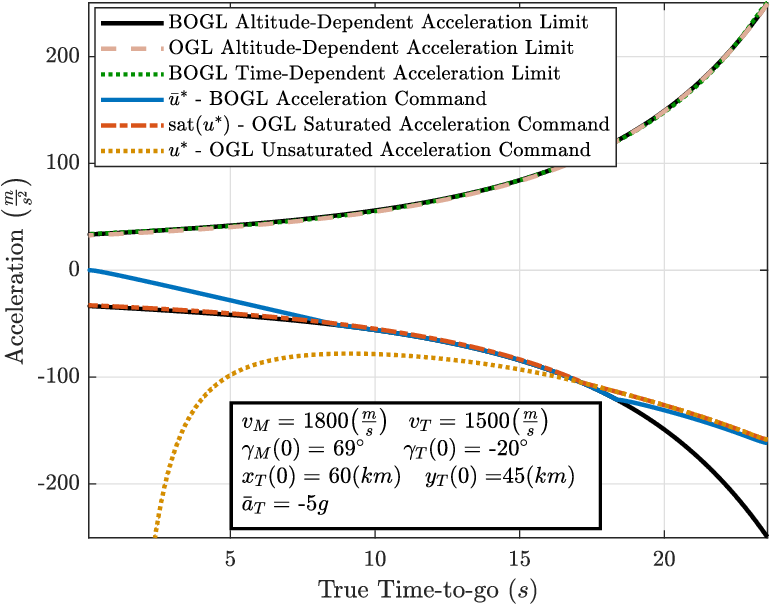}
\caption{\label{AccCommand_SATalt_1stOrder}  Acceleration commands.}
\end{figure}

In \figref{Ztf_SATalt_1stOrder}, the BOGL's $\bar{Z}(t_{f})$, i.e., the expected miss distance during the scenario, is presented along with its OGL counterpart, $Z(t_{f})$. It is evident that the saturated OGL's expected miss distance is smaller (in absolute value) than its BOGL's counterpart at the beginning of the scenario, as the OGL is not aware of the upcoming saturation, and, therefore, predicts a smaller miss distance, resulting in a more moderate maneuver (consistent with the discussion in \secref{proof_BOGL}). However, from the moment its acceleration command curve intersects the acceleration limits, the expected miss distance of the saturated OGL diverges.
This occurs because the OGL does not account for saturation, and, therefore, its expected miss distance calculation is incorrect. Moreover, it does not have enough maneuverability at that stage of the scenario, and, therefore, the expected miss distance continues to grow (in absolute value), leading to a miss. Conversely, it is evident that the BOGL's expected miss distance, i.e. $\bar{Z}(t_{f})$, is small during the whole scenario, even when the missile enters saturation, which implies that the saturation is predicted successfully by the guidance law.
In \figref{ts_SATalt_1stOrder}, the history of the BOGL's predicted switching times-to-go are presented. They do not change much during the scenario, a testament to the good match between the approximated and the actual acceleration limits. Moreover, the black circle in \figref{ts_SATalt_1stOrder} represents the first moment the solutions for the switching times-to-go became complex. This happens towards the end of the scenario, when the absolute value of $\bar{Z}(t_{f})$ decreases (see \figref{Ztf_SATalt_1stOrder}). This results in a smaller ``slope'' of the optimal unsaturated controller, so the unsaturated optimal acceleration command curve does not intersect its bound.
An important point to clarify is that for both of the predicted switching times-to-go, the moments at which they intercept the purple dashed approximated $t_{go}$ line are the moments the missile enters (for $t_{go_{s_{2}}}$) and exits (for $t_{go_{s_{1}}}$) the saturation region. From these moments onward, the value of the corresponding switching time-to-go becomes irrelevant to the guidance law (which can also be seen in \eqref{Ztf_2} and in \eqref{Ztf_1}).

\begin{figure}[hbt!]
\centering
\includegraphics[width=0.57\textwidth]{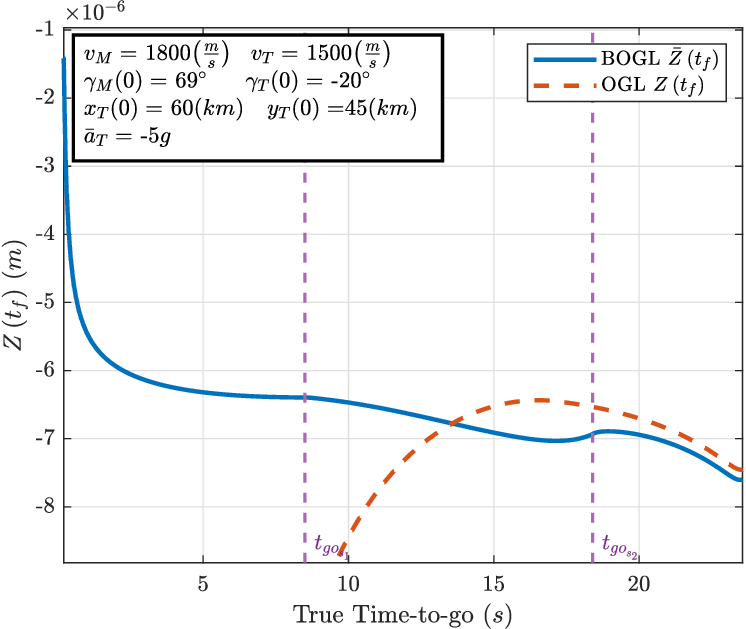}
\caption{\label{Ztf_SATalt_1stOrder} Expected miss distance time history.}
\end{figure}
\begin{figure}[hbt!]
\centering
\includegraphics[width=0.57\textwidth]{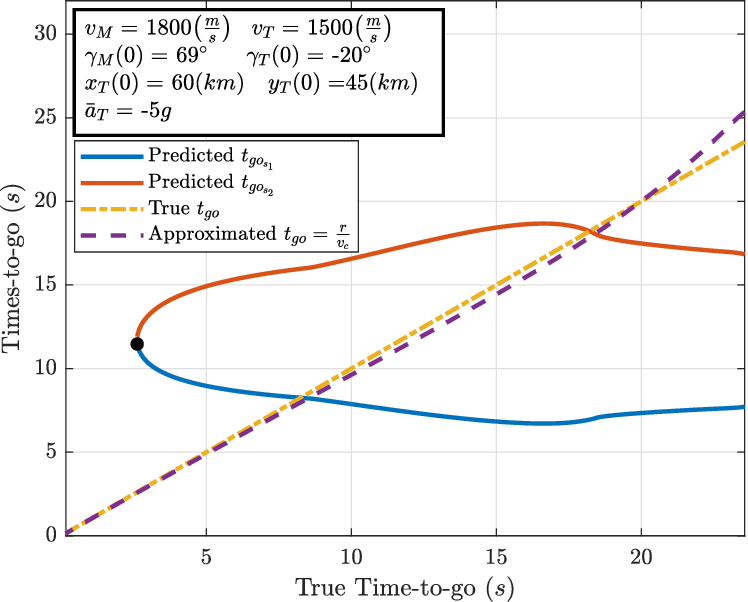}
\caption{\label{ts_SATalt_1stOrder} Expected switching times-to-go histories.}
\end{figure}

In \figref{Miss_distance_Comparison}, capture envelopes are presented for both the BOGL and the OGL, where the criterion of the miss was one meter. Every point on the envelope's boundary represents the simulation conditions for which the missile misses the target by the miss criterion. The horizontal axis is the target's initial altitude, and the vertical axis is its initial flight-path angle. The regions below the curves (i.e., smaller absolute value of the initial target flight-path angle) are the regions in which the miss distance is less than the miss criterion. We call these regions the ``hit regions''. The regions above the curves are the regions in which the miss distance is larger than the miss criterion. We call these regions the ``miss regions''. The solid curves are the BOGL's envelope boundaries, the dashed lines are the OGL's ones, and each color represents a different value of the target maneuver. It is evident that increasing the target acceleration (in absolute value) reduces the hit region. Moreover, for a given target maneuver, the hit region of the BOGL is much greater than its OGL counterpart, which demonstrates the utility of the BOGL.

\begin{figure}[hbt!]
\centering
\includegraphics[width=0.57\textwidth]{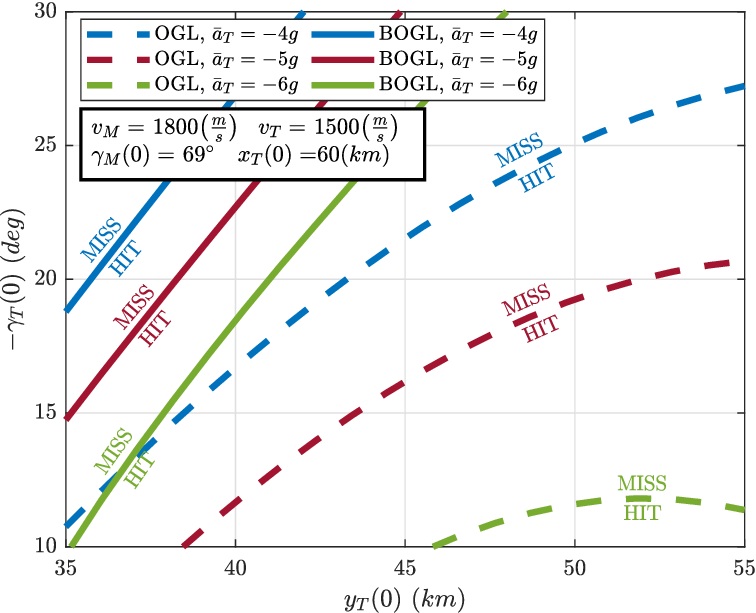}
\caption{\label{Miss_distance_Comparison} Capture envelopes of the BOGL and the OGL.}
\end{figure}

\subsection{Comparison Between the BOGL and the NAPN Guidance Law}
As mentioned in the introduction, the guidance law proposed in \citet{ben2003new} employs an exponential time-to-go weighting of the control-effort running cost and, therefore, tends to command larger accelerations earlier in the engagement, when the missile possesses greater maneuverability. However, unlike the BOGL, the resulting acceleration profile depends on a tuning parameter, $k$, which determines the rate of increase of the weighting function. Similar to the BOGL, this guidance law can successfully intercept the target even in scenarios where the classical OGL fails due to saturation during flight, although appropriate tuning of $k$ is required.
To ensure a fair comparison between the two approaches, the zero-order BOGL is considered, since the guidance law in \cite{ben2003new} was developed for ideal missile dynamics. Accordingly, both guidance laws are simulated using a zero-order dynamic model.
Furthermore, the guidance law in \cite{ben2003new} was originally derived for a non-maneuvering target, although the authors noted that it can be extended to maneuvering-target scenarios. Since the BOGL explicitly accounts for target maneuvers in its derivation, a modified version of the guidance law in \cite{ben2003new} was derived and implemented to incorporate future target acceleration (and, as mentioned earlier, is referred to as NAPN). 

In the BOGL implementation, a parabolic time-to-go approximation was used to estimate the acceleration limitation, and the switching times were computed using the analytically based algorithm presented in \ref{Appendix_Algo} (\algoref{Ztf0_algo}). In the least-squares estimation (\secref{Finding the Parameters}), samples belonging to the saturated region were weighted by a factor of $10^{4}$, to improve the approximation accuracy within the saturation region.
The simulations were performed for both guidance laws using the same parameters presented in \tableref{Parameters} (except for $\tau_{M}$, which is not defined for zero-order dynamics).

Figure \ref{BenAsherComparison_Acceleration} presents the acceleration command histories of the NAPN guidance law for several values of $k$, along with their BOGL counterpart. It can be observed that increasing the value of $k$ results in larger acceleration commands (in absolute value) during the initial phase of the engagement. The corresponding miss distances, control-effort values, and overall cost-function values are summarized in \tableref{AccelerationCommandTable}.
The results indicate that the smallest simulated value of $k$ is insufficient to achieve interception. In this case, the missile does not command sufficiently large accelerations early in the engagement to compensate for the upcoming saturation and, therefore, remains saturated during the terminal phase of the scenario, as reflected by the non-vanishing acceleration command at the end of the engagement. The two larger values of $k$ enable successful interception; however, they require greater control effort than the BOGL and consequently yield larger cost-function values.

\begin{figure}[hbt!]
\centering
\includegraphics[width=0.57\textwidth]{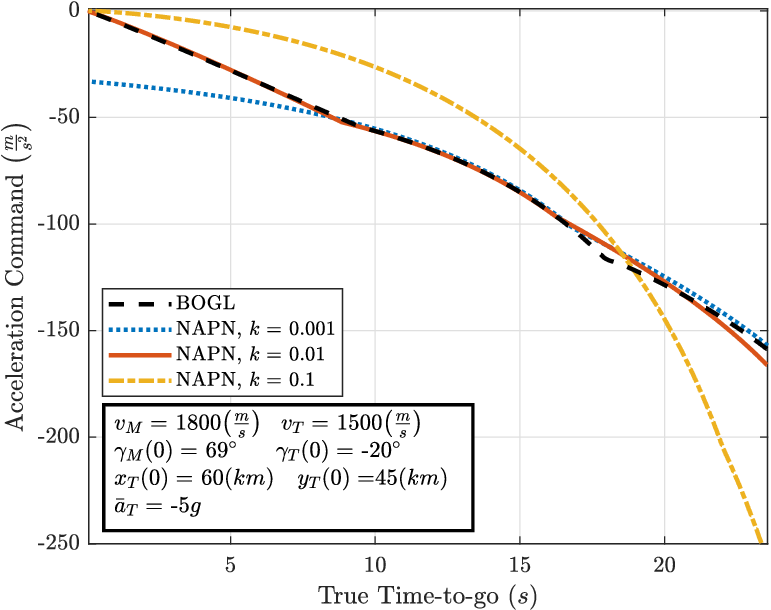}
\caption{\label{BenAsherComparison_Acceleration} Comparison of BOGL and NAPN - acceleration commands.}
\end{figure}

\begin{table}[hbt!]
\centering
\caption{\label{AccelerationCommandTable} Comparison of BOGL and NAPN - sample runs.}
\begin{tabular}{@{}lcccc@{}}
\hline\hline
 & BOGL & NAPN, $k=0.001 \left( \frac{1}{s}\right)$ & NAPN, $k=0.01 \left( \frac{1}{s}\right)$ & NAPN, $k=0.1 \left( \frac{1}{s}\right)$ \\ \hline
Miss $(m)$ & $5.13\times 10^{-5}$ & $119.82$ & $2.49\times 10^{-5}$ & $1.34\times 10^{-4}$ \\
Control Effort $\left(\frac{m^{2}}{s^{3}}\right)$ & $1.668\times 10^{5}$ & $1.679\times 10^{5}$ & $1.67\times 10^{5}$ & $1.923\times 10^{5}$ \\
Cost  $\left(\frac{m^{2}}{s^{3}}\right)$ & $8.343\times 10^{4}$ & $5.999\times 10^{7}$ & $8.351\times 10^{4}$ & $9.622\times 10^{4}$ \\ \hline\hline
\end{tabular}
\end{table}
Figure \ref{BenAsherComparison_CE_gammaM0} presents the NAPN control effort as a function of the weighting parameter $k$. Each circle corresponds to a single simulation run for a specific value of $k$, using the simulation parameters listed in \tableref{Parameters} (excluding $\tau_{M}$). Different colors correspond to different missile launch angles, $\gamma_{M}(0)$. As the launch angle increases, the initial heading error increases as well, and the engagement becomes more challenging for the missile.
A total of 100 logarithmically spaced values of $k$ were simulated, ranging from $10^{-3}(1/s)$ to $10^{-1}(1/s)$. Hollow circles denote missed engagements, whereas filled circles denote successful interceptions, where the miss criterion was taken as one meter.

It can be observed that the scenario with $\gamma_{M}(0)=66^{\circ}$ is relatively easy, as the NAPN-guided missile successfully intercepts the target across all simulated values of $k$, although larger $k$ values require greater control effort. However, for larger launch angles, values of $k$ that are too small lead to a target miss, whereas values that are too large result in excessive control effort. Furthermore, the minimum value of $k$ required for successful interception increases with the launch angle.
Interestingly, the first successful engagements (i.e., the first filled circles) often exhibit lower control effort than the final unsuccessful engagements (i.e., the last hollow circles). This behavior occurs because in successful engagements, the missile exits saturation near the terminal phase of the engagement, enabling successful interception while reducing the accumulated control effort during the terminal portion of the trajectory. The same phenomenon can also be observed in \tableref{AccelerationCommandTable}, where the NAPN with $k=10^{-3}(1/s)$ yields a larger control effort than the case with $k=10^{-2}(1/s)$, since the latter exits saturation near the end of the engagement, as shown in \figref{BenAsherComparison_Acceleration}.

The dashed lines in \figref{BenAsherComparison_CE_gammaM0} represent the BOGL's control effort for the corresponding launch angles. The BOGL hits the target in all scenarios, and it is evident that it provides a lower bound on the control effort across all successful NAPN engagements at every considered launch angle. In other words, no value of $k$ yields successful interception with lower control effort than the BOGL.
Moreover, the optimal value of $k$, in terms of minimizing the control effort, depends on the launch angle and, therefore, on the engagement scenario. This demonstrates the sensitivity of the NAPN tuning process to scenario parameters and highlights a major practical advantage of the BOGL, namely that it does not require iterative parameter tuning.

\begin{figure}[hbt!]
\centering
\includegraphics[width=0.57\textwidth]{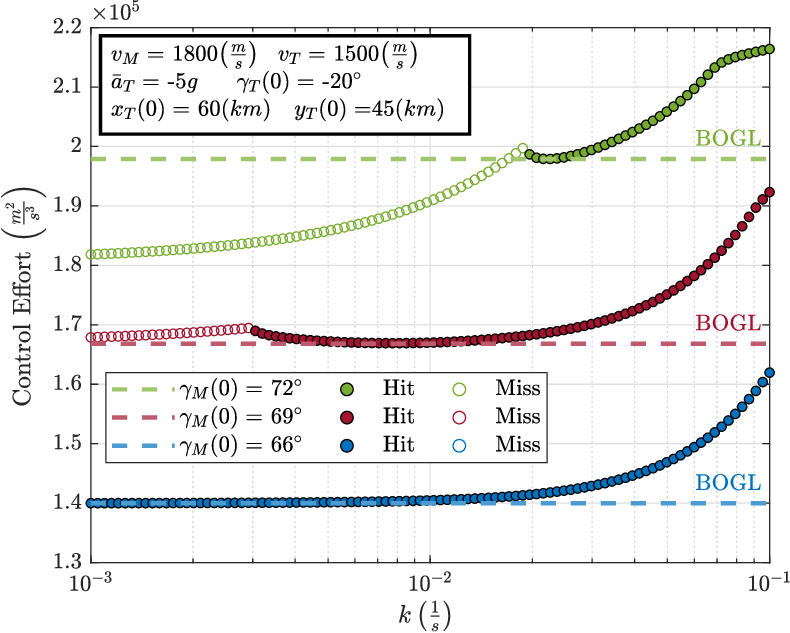}
\caption{\label{BenAsherComparison_CE_gammaM0} Comparison of BOGL and NAPN - control effort.}
\end{figure}

\section{Conclusions} \label{Conclusion}
In this paper, bounded optimal guidance laws for ground-to-air missiles were derived. The guidance problem was formulated as a linear-quadratic optimal-control problem with a bounded controller, considering linear zero-order or first-order strictly proper missile dynamics, together with arbitrary-order linear target dynamics. A hard, time-varying acceleration command constraint was explicitly incorporated into the derivation, thereby forcing the missile to reshape its acceleration profile to meet the saturation limits. This capability stems from the BOGLs’ ability to predict upcoming saturation and maneuver more aggressively while sufficient maneuverability is still available.
Nonlinear simulation results were presented for the first-order BOGL and compared with the corresponding unbounded OGL. The results demonstrated that, in scenarios where the OGL became saturated and failed to intercept the target, the BOGL successfully reshaped its acceleration command and achieved interception even when saturation occurred during a substantial portion of the engagement.
Furthermore, the zero-order BOGL was compared with an equivalent state-of-the-art softly constrained guidance law. The results showed that the softly constrained approach required scenario-dependent tuning, whereas the BOGL provided a lower bound on the control effort among all successful engagements obtained using the softly constrained guidance law, without requiring any tuning, for all considered scenarios.

\appendix
\renewcommand{\thesection}{Appendix \Alph{section}}
\section{Transition Matrix, Reduced Order Dynamics, ZEM, and Integrals} \label{Appendix_Derivation}

\subsection{Transition Matrix - General Order Dynamics} \label{Terms Trasition}
To compute the optimal controller, we need to calculate $\mathbf{D}\mathbf{\Phi}$ (recall \THREEeqsref{Z_dot}{u_opt_free}{Bounded Optimal controller}), which can generally be written as
\begin{equation} \label{DPhi_matrix_general_appx}
      \mathbf{D}\mathbf{\Phi} =
      \left[{\phi}_{11} \quad {\phi}_{12} \quad \bm{\phi}_{1T}^{T} \quad \bm{\phi}_{1M}^{T} \right]
\end{equation}
Solving \eqref{Phi_ode} for each element of $\mathbf{D}\mathbf{\Phi}$ we obtain

\begin{subequations} \label{eq:DPhi_appx}
\begin{align}
     \frac{\dd\phi_{11}}{\dd t_{go}}&=0, \quad \phi_{11}(t_{go}=0)=1 \quad
 \Longrightarrow \quad \phi_{11} = 1 \\ \label{phi_12}
 \frac{\dd\phi_{12}}{\dd t_{go}}&=\phi_{11}=1,
  \quad \phi_{12}(t_{go}=0)=0 \quad \Longrightarrow \quad \phi_{12} =  t_{go} \\\label{d_phi_1T_tgo}
  \frac{\dd \bm{\phi}_{1T}^{T}}{\dd t_{go}}&=\phi_{12} K_{T} \mathbf{c}_{T}^{T}+\bm{\phi}_{1T}^{T} \mathbf{A}_T, \quad \bm{\phi}_{1T}^{T}(t_{go}=0) = \mathbf{0}^{T} \\ \label{d_phi_1M_tgo}
  \frac{\dd \bm{\phi}_{1M}^{T}}{\dd t_{go}}&=-\phi_{12} K_{M} \mathbf{c}_{M}^{T}+\bm{\phi}_{1M}^{T} \mathbf{A}_M , \quad \bm{\phi}_{1M}^{T}(t_{go}=0) = \mathbf{0}^{T}
\end{align}
\end{subequations}
Substituting \eqref{phi_12} into \eqsref{d_phi_1T_tgo} and (\ref{d_phi_1M_tgo}), applying the Laplace transform in the $t_{go}$ domain, rearranging, and taking the inverse transform, yields
\begin{equation}\label{phi_TM_App}
      \bm{\phi}_{1T}^{T} (t_{go})=K_{T} \mathscr{L}_{t_{go}}^{-1} \left[ \frac{1}{s^{2}} \mathbf{c}_{T}^{T}\left(s \mathbf{I}-\mathbf{A}_{T} \right)^{-1} \right], \quad
  \bm{\phi}_{1M}^{T} (t_{go})=-K_{M} \mathscr{L}_{t_{go}}^{-1} \left[ \frac{1}{s^{2}} \mathbf{c}_{M}^{T}\left(s \mathbf{I}-\mathbf{A}_{M} \right)^{-1} \right]
\end{equation}
Substituting \TWOeqsref{eq:DPhi_appx}{phi_TM_App} into \eqref{DPhi_matrix_general_appx} yields \eqref{DPhi}.

\subsection{\texorpdfstring{Calculating $Z(t)$ and the Sign of $\bar{Z}(t_f)$}{Calculating Z(t)}}
\subsubsection{General Solution}
The definition of the ZEM in \eqref{Z} can be written as
\begin{equation} \label{eq:Z_App}
    Z(t) = Z_{g} + Z_{c} , \quad Z_{g}=\mathbf{D}\mathbf{\Phi}\mathbf{x}, \quad Z_{c}= \int_{t}^{t_{f}} \mathbf{D}\mathbf{\Phi} \mathbf{c} w \, \dd \tau 
\end{equation}
Substituting \eqref{DPhi} into \eqref{eq:Z_App} yields
\begin{equation} \label{eq:Zg_App}
Z_g=y+\dot{y}t_{go}+{\bm{\phi}}_{{1}{T}}^{T}\mathbf{x}_{T}+{\bm{\phi}}_{{1}{M}}^{T}\mathbf{x}_{M}
=y+\dot{y}t_{go}+{\bm{\phi}}_{{1}{M}}^{T}\mathbf{x}_{M}+K_T\mathscr{L}_{t_{go}}^{-1}\left[ \frac{1}{s^2}\mathbf{c}_{T}^{T}\left({s}\mathbf{I}-\mathbf{A}_{T}\right)^{-{1}}\right] \mathbf{x}_{T}\\
\end{equation}
Substituting \eqref{DPhi}
into \eqref{eq:Z_App} and using the definition of $G_{T}$ in \eqref{eq:a_T}, yields
\begin{equation} \label{eq:Zc_App}
\begin{split}
Z_c&=\int_{t}^{t_f} \left({\bm{\phi}}_{1T}^{T}{\mathbf{b}}_{T}+K_{T}t_{go} d_{T} \right) w(\tau)\    \dd\tau=\int_{t}^{t_f}\left\{K_T\mathscr{L}_{t_{go}}^{-1}\left[\frac{1}{s^2}\mathbf{c}_{T}^{T}\left({s}\mathbf{I}-\mathbf{A}_{T}\right)^{-1}\mathbf{b}_{T}\right]+K_T d_T t_{go}\right\}w\left(\tau\right)\dd\tau\\
&=\int_{t}^{t_f}{K_T\mathscr{L}_{t_{go}}^{-1}\left[\frac{1}{s^2}\mathbf{c}_T^T\left({s}\mathbf{I}-\mathbf{A}_T\right)^{-{1}}\mathbf{b}_T+\frac{d_T}{s^{2}}\right]w\left(\tau\right)\dd\tau}=K_T\mathscr{L}_{t_{go}}^{-1}\left[\frac{1}{s^2} G_{T}(s) \tilde{w}\left(s\right)\right]
\end{split}
\end{equation}
Substituting \eqsref{eq:Zg_App} and (\ref{eq:Zc_App}) into \eqref{eq:Z_App} and rearranging yields
\begin{equation} \label{eq:Z_general_final_App}
 Z(t) =y+\dot{y}t_{go}+{\bm{\phi}}_{{1}{M}}^{T}\mathbf{x}_{M}+Z_T(t), \quad
 Z_T(t)= K_T\mathscr{L}_{t_{go}}^{-1}\left[ \frac{1}{s^2}\mathbf{c}_{T}^{T}\left({s}\mathbf{I}-\mathbf{A}_{T}\right)^{-{1}} \mathbf{x}_{T}+\frac{1}{s^2} G_{T}(s) \tilde{w}\left(s\right)\right]
\end{equation}

\subsubsection{Zero-Order Missile Dynamics} \label{Optimal Control - Zero Order}
Substituting zero-order missile dynamics into \eqref{phi_TM_App} we obtain
\begin{equation}\label{parameters_0}
    a_{M}=u, \quad \mathbf{A}_{M}=0, \quad \mathbf{b}_{M}=0, \quad \mathbf{c}_{M}^{T}=0, \quad d_{M}=1  \quad \Longrightarrow \quad {\bm{\phi}}_{{1}{M}}^{T}\left(t_{go}\right)=0 
\end{equation}
Substituting \eqref{parameters_0} into \eqref{eq:Z_general_final_App} yields
\begin{equation} \label{eq:Z_0_App}
    Z_0\left(t\right)=y+\dot{y}t_{go}+Z_T(t)
\end{equation}
\subsubsection{First-Order Missile Dynamics} \label{Optimal Control - 1st Order}
Substituting first-order missile dynamics into \eqref{phi_TM_App} we obtain
\begin{equation}\label{parameters_1st}
    a_{M}=x_{M}, \quad \mathbf{A}_{M}=-1/\tau_{M}, \quad \mathbf{b}_{M}=1/\tau_{M}, \quad \mathbf{c}_{M}^{T}=1, \quad d_{M}=0 \quad \Longrightarrow \quad {\bm{\phi}}_{{1}{M}}^{T}\left(t_{go}\right)= 
    -K_{M} \tau_{M}^{2}\psi(\theta)  
\end{equation}
Substituting \eqref{parameters_1st} into \eqref{eq:Z_general_final_App} yields
\begin{equation} \label{phi_1M}
Z_1\left(t\right)=y+\dot{y}t_{go}-K_M\tau_M^2\psi(\theta)a_{M}+Z_T(t)
\end{equation}
\subsubsection{\texorpdfstring{The Sign of $\bar{Z}(t_f)$}{The Sign of Z(t)}}
Substituting the unbounded and the bounded optimal controller structures (\TWOeqsref{u_opt_free}{Bounded Optimal controller}, respectively) into \eqref{Z_dot} yields
\begin{subequations} \label{Z_dot_appendix}
\begin{align}
    &\dot{Z}=-\alpha Z(t_{f})\varphi^{2}\\ 
    &\dot{Z}=-\varphi u_{\max} \text{sat}\left[\frac{\alpha \bar{Z}(t_{f})\varphi}{u_{\max}} \right] = -\left| \varphi \right| u_{\max} \text{sat}\left[\frac{\alpha \bar{Z}(t_{f})\left| \varphi \right|}{u_{\max}} \right]
\end{align}
\end{subequations}
Substituting \eqref{Derivative} into \eqsref{Z_dot_appendix}, we obtain the following ODEs for $Z$ with respect to $t_{go}$
\begin{subequations}
\begin{align}
    &\frac{\dd Z}{\dd t_{go}}=\alpha Z(t_{f})\varphi^{2}, \quad Z(t_{go}=0)=Z(t_{f})\\ 
    & \frac{\dd Z}{\dd t_{go}}= \left| \varphi \right| u_{\max} \text{sat}\left[\frac{\alpha \bar{Z}(t_{f})\left| \varphi \right|}{u_{\max}} \right], \quad Z(t_{go}=0)=\bar{Z}(t_{f})
\end{align}
\end{subequations}
Since the signs of the derivatives are equal to the signs of the initial value at $t_{go}=0$, we obtain that the signs of the expected miss distances of the OGL and the BOGL are the same and equal to the sign of the ZEM, i.e., $\text{sign}\left[Z(t_{f})\right]=\text{sign}\left[\bar{Z}(t_{f})\right]=\text{sign}\left[ Z(t)\right]$.

\subsection{\texorpdfstring{Computing $\varphi$ and $I_{1},I_{2}$}{Computing vatphi, I1, I2}}

\subsubsection{Zero-Order Missile Dynamics}
Substituting the zero-order dynamics in \eqref{parameters_0} into \TWOeqsref{DPhi}{state space terms}, and substituting the result into 
\eqref{Z_dot} yields
 \begin{equation} \label{phi_ideal_app}
    \varphi_0 (t_{go})=
    \begin{bmatrix}
        1 & t_{go} & \bm{\phi}_{1T}^{T} & 0 
      \end{bmatrix} \begin{bmatrix}
        0 & -K_{M} & \left[ \bm{0}\right] & 0 
      \end{bmatrix}^{T}=-K_{M}t_{go}
\end{equation}
Substituting \eqref{phi_ideal_app} into \eqref {integrals} we get the following expressions for $I_{1}$ (with a parabolic acceleration limit) and $I_{2}$
\begin{align} 
I_{1}(t_{go}) =  -K_{M}t_{go}^{2}\left( \frac{\tilde{b}_{2}}{4}t_{go}^{2}+\frac{\tilde{b}_{1}}{3}t_{go}+\frac{\tilde{b}_{0}}{2}\right), \quad
I_{2}(t_{go}) = \frac{K_{M}^{2}}{3}t_{go}^{3}
\end{align}
\subsubsection{First-Order Missile Dynamics}
Substituting the first-order dynamics in \eqref{parameters_1st} into \TWOeqsref{DPhi}{state space terms} and substituting the result into \eqref{Z_dot} yields
 \begin{equation} \label{phi_first_app}
    \varphi_1 (\theta)=       \begin{bmatrix}
        1 & t_{go} & \bm{\phi}_{1T}^{T} & -K_{M}\tau_{M}^{2}\psi(\theta)  
      \end{bmatrix} \begin{bmatrix}
        0 & 0 & \left[ \bm{0}\right] & \frac{1}{\tau_{M}} 
      \end{bmatrix}^{T}=-K_{M}\tau_{M}\psi(\theta)
\end{equation}
Substituting \eqref{phi_first_app} into \eqref{integrals}, and changing variables, we obtain
\begin{equation}
    I_{1}(\theta)= -K_{M}\tau_{M}^{2} \int_{0}^{\theta} \left( \sum_{i=0}^{n} b_{i} \xi^{i}\right) \left( {\mathrm e}^{-\xi}+\xi-1\right) \, \dd\xi ,\quad I_2\left(\theta\right)=\tau_M^3K_M^2\left(-\frac{1}{2}{\mathrm e}^{-2\theta}-2\theta {\mathrm e}^{-\theta}+\frac{\theta^3}{3}-\theta^2+\theta+\frac{1}{2}\right)
\end{equation}
where $I_{1}$ can be solved analytically using integration by parts.
The solutions for second and third-order polynomials are
\begin{subequations} \label{I1_theta}
\begin{align} 
\begin{split}
&I_{1}\left( \theta \right) \big|_{n=2} = -K_{M}\tau_{M}^{2} \left[ \left(1-{\mathrm e}^{-\theta}-\theta+\frac{\theta^{2}}{2}\right) b_{0} + \left(1-\frac{\theta^{2}}{2}+\frac{\theta^{3}}{3}-\theta {\mathrm e}^{-\theta}-{\mathrm e}^{-\theta}\right) b_{1} \right] \\ 
& \qquad \qquad \quad -K_{M}\tau_{M}^{2} b_{2}  \left(2-{\mathrm e}^{-\theta} \theta^{2}-2 \theta {\mathrm e}^{-\theta}-2 {\mathrm e}^{-\theta}-\frac{\theta^{3}}{3}+\frac{\theta^{4}}{4}\right)
\end{split} \\
&I_{1}\left( \theta \right) \big|_{n=3} = I_{1}\left( \theta \right) \big|_{n=2} -K_{M}\tau_{M}^{2}  b_{3}  \left(6-\frac{\theta^{4}}{4}+\frac{\theta^{5}}{5}-{\mathrm e}^{-\theta} \theta^{3}-3 {\mathrm e}^{-\theta} \theta^{2}-6 \theta {\mathrm e}^{-\theta}-6 {\mathrm e}^{-\theta}\right)
\end{align}
\end{subequations}

\section{Sufficient Conditions for at Most Two Switching Times} \label{Appendix_Proof}
The derivation in \secref{Bounded Optimal Guidance Law} assumes a maximum of two intersection points between the absolute value of the unsaturated optimal acceleration command (i.e., $|\alpha \bar{Z}(t_{f}) \varphi|$) and its limit $u_{\max}$. 
Let us derive sufficient conditions to satisfy this assumption using Rolle’s Theorem and two of its Corollaries, which state:
\begin{theorem}[Rolle’s Theorem]
Let $f$ be a continuous function 
on $\left[a,b\right]$, a differentiable function on $\left(a,b\right)$, and $f(a)=f(b)$,
then there exists at least one $c$ in $\left(a,b\right)$ such that $f'(c)=0$.
\end{theorem}
\begin{corollary} \label{corollary1}
Let $f$ be a function that satisfies Rolle’s Theorem conditions, and $f'$ has no more than $n$ roots in $\left(a,b\right)$,
then $f$ has no more than $n+1$ roots in $\left[a,b\right]$.
\end{corollary}
\begin{corollary} \label{corollary2}
Let $f$ be a function that satisfies Rolle’s Theorem conditions and is continuously differentiable $n-1$ times, and its n-th derivative exists in $\left(a,b\right)$. If the n-th derivative has
no roots in $\left(a,b\right)$, then $f$ has no more than $n$ roots in $\left[a,b\right]$.
\end{corollary}
First, for brevity, we use a prime (i.e., $\Box'$) to denote a derivative with respect to $t_{go}$.
The switching times-to-go are the roots of $\ell(\tgo)$ in \eqref{l_def}. According to Corollary \ref{corollary2}, $\ell(\tgo)$ has no more than two roots in $\left[0,\tgo^{initial}\right]\triangleq\bar{I}$ if $\ell''(\tgo)\neq0,\forall\tgo\in \left(0,\tgo^{initial}\right)\triangleq I$. However, this is an over-conservative sufficient condition and depends on the parameters of the problem and the current value of $\bar{Z}(t_{f})$, so it must be checked at every time step.
To improve the computation time and avoid the sensitivity to scenario parameters, we will derive sufficient conditions for $\ell(\tgo)$ to have a maximum of two roots, which are checked once at the beginning of the scenario, and only depend on the properties and coefficients of $\umax(\tgo)$ and the properties of $\varphi (\tgo)$.

\begin{proposition}
$\ell(\tgo)$ has no more than two roots in $\bar{I}$ if $\ell(\tgo)$, $\ell'(\tgo)$, and  $\ell''(\tgo)$ are continuous in $\bar{I}$, differentiable in $I$, and \eqsref{3cond} hold.
\end{proposition}
\begin{proof}
From \eqref{umax_cond} it is evident that $u''_{\max}(\tgo)$ is increasing in $I$, and from \eqref{uopt_cond} it is evident that $\abs{\varphi (\tgo)}''$ is decreasing in $I$. Hence, the second derivative of $\ell(\tgo)$ is increasing in $I$. Therefore, $\ell''(\tgo)$ can be one of five options in $I$: always positive, always negative, always zero, negative and then positive (with a possible intermediate interval for which it equals zero), or zero and then positive. Additionally, the following statement will be useful later on 
\begin{equation} \label{lprime_int}
    \ell'(\tgo)=\ell'(0)+\int_{0}^{\tgo} \ell''(\xi) \, \dd\xi
\end{equation}
The possible combinations for $\ell'(0)$ and $\ell''(\tgo)$ are presented in \tref{Options}, where ${t_{go}}_{1},{t_{go}}_{2}$, and ${t_{go}}_{3}$ are some $\tgo$ values in $I$, and ${t_{go}}_{1} \le {t_{go}}_{2}$.
\begin{table}[hbt!]
\caption{\label{Options} The possible combinations for $\ell'(0)$ and $\ell''(\tgo)$.}
\centering
\begin{tabular}{lccc}
\hline\hline
& $\ell'(0)>0$ & $\ell'(0)<0$ & $\ell'(0)=0$\\\hline
$\ell''(\tgo)>0,\forall\tgo\in I$& Option \#1& Option \#2& Option \#3\\\hline
$\ell''(\tgo)<0,\forall\tgo\in I$& Option \#4& Option \#5& Option \#6\\\hline
$\ell''(\tgo)=0,\forall\tgo\in I$& Option \#7& Option \#8& Option \#9\\\hline
$\ell''(\tgo)<0, \forall \tgo\in \left(0,{t_{go}}_{1}\right)$ & \multirow{3}{*}{Option \#10} & \multirow{3}{*}{Option \#11} & \multirow{3}{*}{Option \#12} \\
$\ell''(\tgo)=0, \forall \tgo\in \left[{t_{go}}_{1},{t_{go}}_{2}\right]$ & & & \\
$\ell''(\tgo)>0, \forall \tgo\in \left({t_{go}}_{2},t_{go}^{initial}\right)$ & & & \\
\hline
$\ell''(\tgo)=0, \quad \forall \tgo\in \left(0,{t_{go}}_{3}\right]$ & \multirow{2}{*}{Option \#13} & \multirow{2}{*}{Option \#14} & \multirow{2}{*}{Option \#15} \\
$\ell''(\tgo)>0, \forall \tgo\in \left({t_{go}}_{3},\tgo^{initial}\right)$ & & & \\
\hline\hline
\end{tabular}
\end{table}

Let us examine each option and conclude the number of roots $\ell(\tgo)$ has in $\bar{I}$ for each one.
First, for options \#1 and \#13 $\ell'(\tgo)$ has no roots in $I$ since it is increasing with an initial condition greater than zero. Similarly, for options \#5 and \#6 $\ell'(\tgo)$ has no roots in $I$ since it is strictly decreasing with a non-positive initial condition. In addition, for options \#7 and \#8 $\ell'(\tgo)$ has no roots in $I$ since $\ell'(\tgo)$ is constant and its initial value is non-zero.
For option \#3 $\ell'(\tgo)$ has no roots in $I$ as well, since it is strictly increasing with a zero initial condition. For all the options mentioned thus far, according to Corollary \ref{corollary2}, $\ell(\tgo)$ has no more than one root in $\bar{I}$, since $\ell'(\tgo)$ has no roots in $I$.

For option \#4 $\ell'(\tgo)$ has no more than one root in $I$, since $\ell'(\tgo)$ is strictly decreasing and its initial condition is positive. Similarly, for options \#2 and \#14 $\ell'(\tgo)$ has no more than one root in $I$, since $\ell'(\tgo)$ is strictly increasing or constant and then strictly increasing, and its initial condition is negative. In addition, it is evident from \eqref{lprime_int} that for options \#11 and \#12 $\ell'(\tgo)$ has no more than one root in $I$ as well, since the initial condition is non-positive and $\ell'(\tgo)$ is decreasing at the beginning of $I$ and then increasing from $t_{go_2}$.
For all those options, Corollary \ref{corollary1} states that $\ell(\tgo)$ has no more than two roots in $\bar{I}$, since $\ell'(\tgo)$ has no more than one root in $I$.

For options \#9 and \#15, $\ell'(\tgo)$ might have an infinite number of roots.
Similarly to \eqref{lprime_int}, the following holds
\begin{equation} \label{l_int}
    \ell(\tgo)=\ell(0)+\int_{0}^{\tgo} \ell'(\xi) \, \dd\xi
\end{equation}
It is evident from \eqref{l_int} that since $\ell'(\tgo)$ is non-negative, and $\ell(0)>0$ (recall \eqref{zero_cond}), $\ell(\tgo)$ has no roots in $\bar{I}$.

The last option left to consider is option \#10, which is divided into two sub-options. The first sub-option is the case in which $\ell'(\tgo)$ is positive at $\tgo=0$, then decreases to zero or a positive number at $\tgo={t_{go}}_{1}$, stays there until $\tgo={t_{go}}_{2}$, and then increases. In that case, $\ell'(\tgo)$ has an infinite number of roots or no roots. Yet, similarly to options \#9 and \#15, it is evident from \eqref{l_int} that since $\ell'(\tgo)$ is non-negative, and since $\ell(0)>0$, $\ell(\tgo)$ has no roots in $\bar{I}$.
The second sub-option is the case in which $\ell'(\tgo)$ is positive at $\tgo=0$, then decreases until it becomes negative, stays at this negative value, and then increases and becomes positive again. In that case $\ell'(\tgo)$ has two roots in $I$.
Let us denote those roots as $t_{go_{r_{1}}}$ and $t_{go_{r_{2}}}$, where $t_{go_{r_{1}}}<t_{go_{r_{2}}}$. 
It is evident that $\ell'(\tgo)>0,\forall\tgo\in\left(0,t_{go_{r_{1}}}\right)$, and recalling \eqref{zero_cond} we obtain $\text{sign}\left[\ell(0)\right]=\text{sign}\left[\ell'(\tgo)\right]$ in that interval.
Thus, it is evident from \eqref{l_int} that $\ell(\tgo)$ has no roots in $\tgo\in\left[0,t_{go_{r_{1}}}\right]$. Since $\ell'(\tgo)$ has just one root in $\tgo\in\left(t_{go_{r_{1}}},\tgo^{initial}\right)$, Corollary \ref{corollary1} states that $\ell(\tgo)$ has no more than two roots in that closed interval. Therefore, $\ell(\tgo)$ has no more than two roots in $\bar{I}$.

We can therefore conclude that for all possible cases, if the conditions in \eqsref{3cond} are satisfied, $\ell(\tgo)$ has no more than two roots in $\bar{I}$, which proves the claim.
\end{proof}

\section{Switching Times-to-go and Terminal ZEM for Zero-Order Missile Dynamics} \label{Appendix_Algo}
In this Appendix, the algorithm for finding the switching times-to-go and $\bar{Z}(t_{f})$ is presented for a zero-order missile and a parabolic acceleration limit. This combination enables an analytical solution for the switching times-to-go given $\bar{Z}(t_{f})$, thereby simplifying the algorithm and improving convergence time. Moreover, since this algorithm converges substantially faster, its result is used as the initial guess in the general algorithm presented in \ref{Algo_SwitchingTimes}. First, the algorithm is presented, and then an analysis of the maximal and minimal values of $\bar{Z}(t_{f})$ in the flag=3 region is performed.

\subsection{The Algorithm}
The switching times-to-go are the intersection points between the optimal unsaturated acceleration command and its time-to-go-dependent limit, that is
\begin{equation} \label{intercession zero}
\tilde{b}_{2}t_{go_{s}}^{2}+\tilde{b}_{1}t_{go_{s}}+\tilde{b}_{0} =\alpha |\bar{Z}(t_{f}) K_M| t_{go_{s}} 
\end{equation}
This is a second order polynomial in the switching times-to-go, which can be solved as functions of $\bar{Z}(t_{f})$
\begin{equation} \label{tgos_0}
t_{go_{s_{1,2}}} = \frac{1}{2\tilde{b}_{2}}\left( -\tilde{b}_{Z}\mp\sqrt{\tilde{b}_{Z}^{2}-4\tilde{b}_{2}\tilde{b}_{0}}\right), \quad    \tilde{b}_{Z} \dfn \tilde{b}_{1}-\alpha\left| \bar{Z}(t_{f})K_{M}\right|
\end{equation}
Assuming the value of flag for a certain time-to-go is known, one can substitute \eqref{tgos_0} into \TWOeqsref{Ztf_3}{Ztf_2} and calculate $\bar{Z}(t_{f})$ numerically for flag=3 and flag=2, respectively. For flag=1 the missile is unsaturated and the expression for $\bar{Z}(t_{f})$ is independent of the switching times-to-go (recall \eqref{Ztf_1}).
The implemented algorithm is presented in \algoref{Ztf0_algo}.
First, the algorithm searches for intersection points between the unbounded optimal controller and its limits. If there are no such points, or if the times-to-go at which the intersections occur are greater than the current value of the time-to-go, then the missile is at flag=1, and the BOGL reduces to the OGL.
If the missile is not in the flag=1 region, the algorithm checks if it is in the flag=3 region.
If, at a given $\tgo$, the missile is at flag=3, the predicted miss distance at the end of the scenario must satisfy
\begin{equation} \label{flag3_cond}
    \mathcal{Z}_{\min}\le\abs{\bar{Z}(t_{f})}_{\text{flag}(t_{go})=3}\le\mathcal{Z}_{\max}, \quad
       \mathcal{Z}_{\min} = \frac{\tilde{b}_{1}+2\sqrt{\tilde{b}_{2}\tilde{b}_{0}}}{\alpha\abs{K_{M}}}, \quad \mathcal{Z}_{\max} = \frac{\tilde{b}_{2}\tgo^{2}+\tilde{b}_{1}\tgo+\tilde{b}_{0}}{\alpha\abs{K_{M}}\tgo} 
\end{equation}
The expressions in \eqref{flag3_cond} are obtained from the definition of the flag=3 region, where $\abs{\bar{Z}(t_{f})}<\mathcal{Z}_{\min}$ means the missile is in the flag=1 region, and $\abs{\bar{Z}(t_{f})}>\mathcal{Z}_{\max}$ means the missile is in the flag=2 region.
The full derivation is presented in the next subsection to improve readability.
Similarly to the analysis presented in \secref{Sec:SolutionExistence}, the corrections the missile can achieve using the optimal controller with $\mathcal{Z}_{\min}$ and $\mathcal{Z}_{\max}$ are obtained by substituting them into \TWOeqsref{Corflag1}{2}, respectively
\begin{equation} \label{Delta_Zmin_Zmax}
        \Delta Z_{\min} = \alpha \mathcal{Z}_{\min} I_{2}(\tgo), \quad
        \Delta Z_{\max} = \alpha \mathcal{Z}_{\max}I_{2}\left(t_{go_{s_{1}}}^{*}\right)+\abs{I_{1}(\tgo)-I_{1}\left(t_{go_{s_{1}}}^{*}\right)}, \quad 
        t_{go_{s_{1}}}^{*} = \frac{\tilde{b}_{0}}{\tilde{b}_{2} \tgo}
\end{equation}
where $t_{go_{s_{1}}}^{*}$ is obtained by substitution of $\mathcal{Z}_{\max}$ in \eqref{flag3_cond} as $\bar{Z}(t_{f})$ into $t_{go_{s_{1}}}$ in \eqref{tgos_0}.
If the missile is indeed in the flag=3 region, \eqref{flag3_cond} implies that the following conditions must hold
\begin{subequations} \label{flag3_cond_algo}
\begin{align}
    &\abs{Z_0(t)}_{\text{flag}(t_{go})=3}\le \mathcal{Z}_{\max}+\Delta Z_{\max}  \label{Delta_Zmax} \\
    &t_{go_{s_{1}}}^{*}<\tgo \label{tgo_s1} \\
   &\textcolor{black}{\abs{Z_0(t)}_{\text{flag}(t_{go})=3}\ge \mathcal{Z}_{\min}+\Delta Z_{\min}}  \label{Delta_Zmin}
    \end{align}
\end{subequations}
Since the algorithm first checks if the missile is in the flag=1 region, the condition in \eqref{Delta_Zmin} is satisfied automatically when checking if the missile is in the flag=3 region.
If the missile is not in the flag=1 region or in the flag=3, it is, of course, in the flag=2 region, in which $\abs{\bar{Z}(t_{f})}_{\text{flag}(t_{go})=2}>\mathcal{Z}_{\max}$.

The numerical method we used to calculate $\bar{Z}(t_{f})$ is the Newton–Raphson method, in which the algorithm gets an initial guess, and the update equation is given by
\begin{equation} \label{Newton-Raphson_eqn}
    \bar{Z}(t_{f})_{k+1} = \bar{Z}(t_{f})_{k} - \frac{F_{i}\left[\bar{Z}(t_{f})_{k} \right]}{F'_{i}\left[\bar{Z}(t_{f})_{k} \right]}, \quad F_{i}[\bar{Z}(t_{f})] = \bar{Z}(t_{f})D_{i}[\bar{Z}(t_{f})]-N_{i}[\bar{Z}(t_{f})], \quad i\in\left\{2,3 \right\}
\end{equation}
where $i$ is the value of flag, and $F_{2}, F_{3}$ are derived from \TWOeqsref{Ztf_2}{Ztf_3}, respectively, and
\begin{subequations}
    \begin{align}
        &N_2[\bar{Z}(t_{f})]=Z(t)+\sign\left[ Z(t)\right]\left[ I_{1}(t_{go})-I_{1}(t_{go_{s_{1}}})\right] , \,
        D_2[\bar{Z}(t_{f})] =1+\alpha I_{2}(t_{go_{s_{1}}}) \label{F2}\\
        &N_3[\bar{Z}(t_{f})]=Z(t)+\sign\left[ Z(t)\right]\left[ I_{1}(t_{go_{s_{2}}})-I_{1}(t_{go_{s_{1}}})\right], \,
         D_3[\bar{Z}(t_{f})]=1+\alpha\left[ I_{2}(t_{go})+I_{2}(t_{go_{s_{1}}})-I_{2}(t_{go_{s_{2}}})\right] \label{F3}
    \end{align}
\end{subequations}
where $t_{go_{s_{1}}},t_{go_{s_{2}}}$ are also functions of $\bar{Z}(t_{f})$, as can be seen in \eqref{tgos_0}. The derivatives of $F_{2},F_{3}$ in \eqref{Newton-Raphson_eqn} with respect to $\bar{Z}(t_{f})$, i.e. $F'_{2}$ and $F'_{3}$, are
computed analytically, and not presented explicitly to improve readability.
The Newton–Raphson algorithm is presented in \algoref{Newton-Raphson_algo}, where the terms $\Delta t_{go_{s_{1}}},\Delta t_{go_{s_{2}}}$ are the differences between the first and the second switching times-to-go between two consecutive iterations, respectively, and $t_{go_{s_{tol}}}$ is the convergence termination criteria on $\Delta t_{go_{s_{1}}},\Delta t_{go_{s_{2}}}$ between iterations.
\begin{algorithm}[hbt!]
\caption{Calculating $\bar{Z}(t_{f})$ for Zero-Order Missile Dynamics and a Parabolic Acceleration Limit.}\label{Ztf0_algo}
    \begin{algorithmic}[1]
        \State Compute $\bar{Z}(t_{f})$ assuming flag=1 (\eqref{Ztf_1}).
        \State Find $t_{go_{s_{1}}}\left[\bar{Z}(t_{f}) \right],t_{go_{s_{2}}}\left[\bar{Z}(t_{f}) \right]$ (\eqref{tgos_0}).
        \If {$t_{go_{s_{1}}}$ has a complex value or $t_{go_{s_{1}}}>t_{go}$}
            \State The missile is in the flag=1 region.
            \State Return $\bar{Z}(t_{f})$.
        \Else
            \State Compute $\mathcal{Z}_{\min}$ and  $\mathcal{Z}_{\max}$ (\eqref{flag3_cond}) and $\Delta Z_{\max}$
            (\eqref{Delta_Zmin_Zmax}).
            \If {the conditions for flag=3 in \TWOeqsref{Delta_Zmax}{tgo_s1} are satisfied}
                \State Compute $\bar{Z}(t_{f})$ assuming flag=3 using \algoref{Newton-Raphson_algo}.  
                \State Return $\bar{Z}(t_{f})$.
            \Else
            \State Compute $\bar{Z}(t_{f})$ assuming flag=2 using \algoref{Newton-Raphson_algo}.
            \State Return $\bar{Z}(t_{f})$.
            \EndIf
        \EndIf
    \end{algorithmic}
\end{algorithm}
\begin{algorithm}[hbt!]
\caption{Newton-Raphson for Zero-Order Missile Dynamics and a Parabolic Acceleration Limit.}\label{Newton-Raphson_algo}
    \begin{algorithmic}[1]
        \If {first time step}
        \State Set $\bar{Z}(t_{f})_{0}$ as the $\bar{Z}(t_{f})$ from the unbounded case (\eqref{Ztf_1}).
        \Else
        \State Set $\bar{Z}(t_{f})_{0}$ as the $\bar{Z}(t_{f})$ from the previous time step.
        \EndIf
        \State Compute the values of $t_{go_{s_{1}}},t_{go_{s_{2}}}$ using $\bar{Z}(t_{f})_{0}$ (\eqref{tgos_0}).
        \State Compute $\bar{Z}(t_{f})_{1}$ using \eqref{Newton-Raphson_eqn}.
        \State Recompute the values of $t_{go_{s_{1}}},t_{go_{s_{2}}}$ using $\bar{Z}(t_{f})_{1}$ (\eqref{tgos_0}).
        \While {$\left|\Delta t_{go_{s_{1}}} \right|,\left|\Delta t_{go_{s_{2}}} \right|>t_{go_{s_{tol}}}$}
            \State Compute $\bar{Z}(t_{f})_{k+1}$ using \eqref{Newton-Raphson_eqn}.
            \State Recompute $t_{go_{s_{1}}},t_{go_{s_{2}}}$ according to the updated $\bar{Z}(t_{f})_{k+1}$ value (\eqref{tgos_0}).
        \EndWhile
        \State Return $\bar{Z}(t_{f})$.
    \end{algorithmic}
\end{algorithm}

\subsection{\texorpdfstring{Maximal and Minimal Values of $\left|\bar{Z}(t_{f}) \right|$ in the flag=3 Region}{Maximal and Minimal Values}}\label{Maximal and Minimal Values}
In this subsection, we find the maximal and the minimal values of $\abs{\bar{Z}(t_{f})}$ for the flag=3 region. 
They are used to check if the missile is in that region, as can be seen in \algoref{Ztf0_algo}. 
Let us assume, for a certain value of $\tgo$ denoted by $\tgo^{*}$, that the missile is in the flag=3 region. It is evident from \eqref{Bounded Optimal Controller - 0th} that a smaller value of $\abs{\bar{Z}(t_{f})}$ leads to a smaller slope of the unsaturated optimal controller line. Therefore, a too small value of $\abs{\bar{Z}(t_{f})}$ may cause zero intersection points between the unsaturated optimal controller line and the acceleration limit curve in $\tgo\in(0,\tgo^{*}]$. In such a case, the missile is currently in the flag=1 region, which is not the assumed case, as mentioned before. Similarly, a too large value of $\abs{\bar{Z}(t_{f})}$ may lead to a single intersection point in $(0,\tgo^{*}]$, and another one in $(\tgo^{*},\infty)$. In such a case, the missile is currently in the flag=2 region, which is also not the assumed current case. Hence, if the missile is in the flag=3 region, there are $\mathcal{Z}_{\min},\mathcal{Z}_{\max}$ for which $\abs{\bar{Z}(t_{f})}\in\left[\mathcal{Z}_{\min},\mathcal{Z}_{\max}\right]$.
A qualitative graphical representation of the above is presented in \figref{Zmax_Zmin}.

\begin{figure}[hbt!]
\centering
\includegraphics[width=0.57\textwidth]{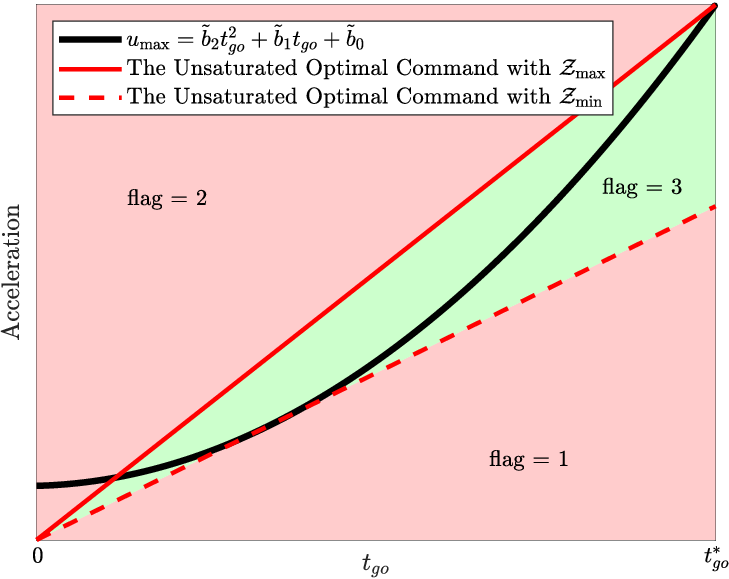}
\caption{\label{Zmax_Zmin} A qualitative graphical representation of $\mathcal{Z}_{\min},\mathcal{Z}_{\max}$.}
\end{figure}

Since the switching times-to-go must be positive, \eqref{tgos_0} implies that $\tilde{b}_{Z}=\tilde{b}_{1}-\alpha \abs{Z(t_{f})K_{M}}<0$. 
In addition, as mentioned in \secref{Finding the Parameters}, $\tilde{b}_{2}>0$ so the acceleration limit approximation represents the convex behavior of the real, altitude-dependent, acceleration limit.
Moreover, since $\tilde{b}_{0}>0$, it is evident that real solutions to the switching times-to-go exist if and only if $\abs{\tilde{b}_{Z}} \ge 2\sqrt{\tilde{b}_{2}\tilde{b}_{0}}$.
Since $\tilde{b}_{Z}$ must be negative, we obtain that a necessary condition for the missile to be currently in the flag=2 or flag=3 regions is $\abs{\bar{Z}(t_{f})}\ge\mathcal{Z}_{\min}$, where $\mathcal{Z}_{\min}$ is presented in \eqref{flag3_cond}.

In the flag=3 region, the maximal possible value of $\tgosTwo$ is $\tgo^{*}$. Moreover, it is evident from \eqref{tgos_0} that a larger value of $\abs{\bar{Z}(t_{f})}$ yields a larger value of $\tgosTwo$ (a larger slope). That is, the maximal value of $\abs{\bar{Z}(t_{f})}$ in flag=3 region is obtained by substituting $\tgosTwo=\tgo^{*}$ into \eqref{intercession zero} and solve for $\abs{\bar{Z}(t_{f})}$.
Hence, a necessary condition for the missile to be in the flag=3 region is $\abs{\bar{Z}(t_{f})}\le\mathcal{Z}_{\max}$, where $\mathcal{Z}_{\max}$ is presented in \eqref{flag3_cond}.
In conclusion, if the missile is in the flag=3 region, the absolute value of the expected miss distance at the end of the scenario,
must satisfy $\abs{\bar{Z}(t_{f})} \in \left[\mathcal{Z}_{\min},\mathcal{Z}_{\max}\right]$.

\section*{Acknowledgment}
This work was supported by the PMRI – Peter Munk Research Institute - Technion. 

The authors acknowledge the use of ChatGPT and Grammarly for editorial assistance during manuscript preparation, including grammar correction and language refinement.

\bibliography{sample}

\end{document}

%% file: command.tex
\newcommand{\dfn}{\triangleq}

\newcommand{\dd}{\text{d}}
 
\newcommand{\sign}{\text{sign}}

\newcommand{\tgo}{t_{go}}
\newcommand{\umax}{u_{\max}}

\newcommand{\tgosTwo}{t_{go_{s_{2}}}}

\newcommand{\abs}[1]{\left|{#1} \right|}

\renewcommand{\eqref}[1]{Eq.~(\ref{#1})}
\newcommand{\tableref}[1]{Table~(\ref{#1})}
\newcommand{\eqsref}[1]{Eqs.~(\ref{#1})}

\newcommand{\figref}[1]{Fig.~\ref{#1}}

\newcommand{\secref}[1]{Sec.~\ref{#1}}

\newcommand{\THREEeqsref}[3]{Eqs.~(\ref{#1}), (\ref{#2}), and (\ref{#3})}
\newcommand{\TWOeqsref}[2]{Eqs.~(\ref{#1}) and (\ref{#2})}
\newcommand{\tref}[1]{Table.~\ref{#1}}
\newcommand{\algoref}[1]{Algorithm ~\ref{#1}}

\newcommand{\comment}[1]{}